%% file: graphons.tex
\documentclass[journal]{IEEEtran}

\usepackage[english]{babel}

\usepackage{ifpdf}

\usepackage{cite} 
\usepackage{url}
\usepackage{hyperref}

\ifCLASSINFOpdf
	\usepackage[pdftex]{graphicx}
	\graphicspath{{./figures/}}
\else
	\usepackage[dvips]{graphicx}
	\graphicspath{./figures/}
\fi
\usepackage{color}
\usepackage{pgf, tikz, pgfplots}
\usetikzlibrary{shapes, arrows, automata}
\usetikzlibrary{calc,hobby,decorations}

\usepackage[cmex10]{amsmath}
\usepackage{amsfonts, amssymb, amsthm}
\usepackage{mathrsfs}

\usepackage{adjustbox}

\usepackage{array}
\usepackage{enumerate}
\usepackage{multirow}
\usepackage{rotating}
\usepackage{subcaption}
\usepackage{dsfont}

\input{my_sections.tex}

\input{mysymbol.sty}

\definecolor{penndarkestblue}{cmyk}{1,0.74,0,0.77}
\definecolor{penndarkerblue}{cmyk}{1,0.74,0,0.70}
\definecolor{pennblue}{cmyk}{0.99,0.66,0,0.57} 
\definecolor{pennlighterblue}{cmyk}{0.98,0.44,0,0.35}
\definecolor{pennlightestblue}{cmyk}{0.38,0.17,0,0.17} 

\definecolor{penndarkestred}{cmyk}{0,1,0.89,0.66}
\definecolor{penndarkerred}{cmyk}{0,1,0.88,0.55}
\definecolor{pennred}{cmyk}{0,1,0.83,0.42} 
\definecolor{pennlighterred}{cmyk}{0,1,0.6,0.24}
\definecolor{pennlightestred}{cmyk}{0,0.43,0.26,0.12} 

\definecolor{penndarkestgreen}{cmyk}{1,0,1,0.68}
\definecolor{penndarkergreen}{cmyk}{1,0,1,0.57}
\definecolor{penngreen}{cmyk}{1,0,1,0.44} 
\definecolor{pennlightergreen}{cmyk}{1,0,1,0.25}
\definecolor{pennlightestgreen}{cmyk}{0.43,0,0.43,0.13}

\definecolor{penndarkestorange}{cmyk}{0,0.65,1,0.49}
\definecolor{penndarkerorange}{cmyk}{0,0.65,1,0.33}
\definecolor{pennorange}{cmyk}{0,0.54,1,0.24} 
\definecolor{pennlighterorange}{cmyk}{0,0.32,1,0.13}
\definecolor{pennlightestorange}{cmyk}{0,0.15,0.46,0.06}
	
\definecolor{penndarkestpurple}{cmyk}{0,1,0.11,0.86}
\definecolor{penndarkerpurple}{cmyk}{0,1,0.13,0.82}
\definecolor{pennpurple}{cmyk}{0,1,0.11,0.71} 
\definecolor{pennlighterpurple}{cmyk}{0,1,0.05,0.46}
\definecolor{pennlightestpurple}{cmyk}{0,0.35,0.02,0.23}
	
\definecolor{pennyellow}{cmyk}{0,0.20,1,0.05} 
\definecolor{pennlightgray1}{cmyk}{0,0,0,0.05}
\definecolor{pennlightgray3}{cmyk}{0.01,0.01,0,0.18}
\definecolor{pennmediumgray1}{cmyk}{0.04,0.03,0,0.31}
\definecolor{pennmediumgray4}{cmyk}{0.08,0.06,0,0.54}
\definecolor{penndarkgray2}{cmyk}{0.09,0.07,0,0.71}
\definecolor{penndarkgray4}{cmyk}{0.1,0.1,0,0.92}

\def\Tr{\mathsf{T}}
\def\Hr{\mathsf{H}}

\renewcommand{\blue}{\color{black}}
\renewcommand{\red}{\color{black}}

\newcommand{\vertiii}[1]{{\left\vert\kern-0.25ex\left\vert\kern-0.25ex\left\vert #1 \right\vert\kern-0.25ex\right\vert\kern-0.25ex\right\vert}}

\newtheorem{lemma}{\hspace{0pt}\bf Lemma}
\newtheorem{proposition}{\hspace{0pt}\bf Proposition}

\newtheorem{theorem}{\hspace{0pt}\bf Theorem}
\newtheorem{corollary}{\hspace{0pt}\bf Corollary}

\newtheorem{remark}{\hspace{0pt}\bf Remark}

\newtheorem{definition}{\hspace{0pt}\bf Definition}

\newtheorem{innercustomgeneric}{\customgenericname}
\providecommand{\customgenericname}{}
\newcommand{\newcustomtheorem}[2]{%
  \newenvironment{#1}[1]
  {%
   \renewcommand\customgenericname{#2}%
   \renewcommand\theinnercustomgeneric{##1}%
   \innercustomgeneric
  }
  {\endinnercustomgeneric}
}

\newcustomtheorem{customthm}{Theorem}
\newcustomtheorem{customlemma}{Lemma}

\begin{document}

\title{Graphon Signal Processing}

\author{Luana~Ruiz, Luiz~F.~O.~Chamon~
        and~Alejandro~Ribeiro
\thanks{This work in this paper was supported by NSF CCF 1717120, ARO W911NF1710438, ARL DCIST CRA W911NF-17-2-0181, ISTC-WAS and Intel DevCloud. Preliminary results have been accepted for publication at the ICASSP20 conference \cite{ruiz2019graphon}. L. Ruiz, L. F. O. Chamon and A. Ribeiro are with the Dept. of Electrical and Systems Eng., Univ. of Pennsylvania.  
}
}

\markboth{IEEE TRANSACTIONS ON SIGNAL PROCESSING (ACCEPTED)}%
{Graphon Filters}

\maketitle

\begin{abstract}

Graphons are infinite-dimensional objects that represent the limit of convergent sequences of graphs as their number of nodes goes to infinity. This paper derives a theory of graphon signal processing centered on the notions of \textit{graphon Fourier transform} and linear shift invariant \textit{graphon filters}, the graphon counterparts of the graph Fourier transform and graph filters. It is shown that for convergent sequences of graphs and associated graph signals: (i) the graph Fourier transform converges to the graphon Fourier transform when the graphon signal is bandlimited; (ii) the spectral and vertex responses of graph filters converge to the spectral and vertex responses of graphon filters with the same coefficients. These theorems imply that for graphs that belong to certain families, i.e., that are part of sequences that converge to a certain graphon, graph Fourier analysis and graph filter design have well defined limits. In turn, these facts extend applicability of graph signal processing to graphs with large number of nodes --- since signal processing pipelines designed for limit graphons can be applied to finite graphs --- and to dynamic graphs --- since we can relate the result of SP pipelines designed for different graphs from the same convergent graph sequence.
\end{abstract}


\IEEEpeerreviewmaketitle


\section{Introduction} \label{sec:intro}

\input{intro-graphon.tex}


\input{graphs_and_graphons.tex}


\section{Graphon Signal Processing} \label{sec:gft}

\input{gft-graphon.tex}


\section{GSP converges to WSP} \label{sec:filters}

\input{filter-graphon.tex}


\section{Numerical Experiments} \label{sec:sims}

\input{sims-graphon.tex}





\section{Conclusions} \label{sec:conclusions}

\input{conclusions-graphon.tex} 


\bibliographystyle{IEEEtran}
\bibliography{myIEEEabrv,bib-graphon}

\appendices

\input{app-graphon.tex}

\end{document}

%% file: my_sections.tex
\usepackage{needspace}



%% file: intro-graphon.tex


Graph signal processing (GSP) provides an array of tools to process signals supported on graphs \cite{shuman13-mag, sandryhaila14-freq, ortega2018graph} but suffers from limitations in the case of graphs with {\it large} number of nodes or {\it dynamic} topologies. In these cases, just the acquisition of the graph may be challenging, which hinders the use of GSP tools such as filtering \cite{segarra17-linear, sandryhaila2013discrete} and graph neural network design \cite{gama18-gnnarchit} because they take the graph structure as a given. Other GSP tools like sampling \cite{rui2016dimensionality, chen2015discrete, marques2015sampling} deal precisely with acquiring compact representations of graph signals. However, the {\it design} of sampling sets \cite{chamon2017greedy} requires not only access to the graph but the computation of an eigendecomposition that can be very costly for large matrices \cite{morgan1991computing}, \blue{\cite[Chapter 1.1]{paige1971computation}}. Challenges are most acute when the graph is both large {\it and} dynamic. In such cases, costly numerical computations must, in principle, be repeated as the graph changes, because the effect of graph perturbations is understood only in the case of relabelings \cite{ruiz19-inv} or small perturbations that induce small changes on the original eigenspace \cite{gama2019stability}.

Yet, large graphs can often be identified as being similar to each other in the sense that they share structural properties. For instance, Figs. \ref{fig:experiments}\subref{small_graph1}-\subref{small_graph2} show two instances of a random graph with $20$ nodes, and Fig. \ref{large_graph} a random graph with $50$ nodes. These graphs look similar and one can therefore foresee that analyzing signals supported on either of them should yield similar results. If this were the case, it would mitigate the challenge of dynamic variation since we could then design a filter for the graph in Fig. \ref{small_graph1} and use it in the graph in Fig. \ref{small_graph2}. Similarly, it would mitigate the challenge of large size because we could design a filter for the graph in Fig. \ref{small_graph1} and use it to process signals supported on the graph in Fig. \ref{large_graph}. 
This paper formalizes this intuition ans shows that this graph interchangeability is possible when the graphs belong to the same ``family'',
where each family is identified by a different \textit{graphon}; see Fig. \ref{er_graphon}.

Graphons can be thought of as the infinite-dimensional counterparts of graphs, i.e., as graphs with an uncountable number of nodes. Appearing in many disciplines, they have been used to estimate random graph models in mathematics and statistics\cite{wolfe2013nonparametric, airoldi2013stochastic, de2017adaptive, xu2017rates, gao2015rate, xu2014edge}; stabilize large-scale networks of linear systems in controls \cite{gao2018graphon}; and perform graph partitioning \cite{rohe2011spectral, diao2016model}, node centrality \cite{avella2018centrality} and network game equilibria computations \cite{parise2019graphon} in very large networks. 
Graphons have two theoretical interpretations. They can be seen as generative models for families of graphs with weighted or stochastic edges \cite[Chapter 10]{lovasz2012large},
and as the limit objects of convergent sequences of graphs \cite[Chapter 7]{lovasz2012large},\cite{lovasz2006limits}. In practice, these two interpretations suggest that graphons identify families of networks that are similar in the sense that the density of certain ``motifs'' is preserved. This motivates the study of signal processing on graphons as a way to enable the analysis of signals supported on large and/or dynamic graphs.

In this work, we thus introduce graphon signal processing (WSP), a framework to synthesize, analyze and process signals on graphons. More specifically, we put forward three novel technical contributions:
(i) we define graphon signals and their graphon Fourier transforms (Def. \ref{defn:wft}), which can be seen as the continuous counterparts of graph signals and of their graph Fourier transforms; (ii) we show, by building upon the results of \cite{morency2017signal}, that
the graph Fourier transform converges to the graphon Fourier transform (Thm. \ref{thm:wft}) when the graphon signal is bandlimited (Def. \ref{defn:bandlimited}); (iii) we define linear-shift-invariant (LSI) graphon filters (Def. \ref{defn:linear}), and prove that LSI graph filters converge to LSI graphon filters in both the spectral (Thm. \ref{thm:transfer_fcn}) and vertex domains (Thms. \ref{thm:non_disc}--\ref{thm:non_disc2}). 

%
\begin{figure*}[t]
	\centering
	\begin{subfigure}{0.23\textwidth}
		\centering
		\includegraphics[height=2.8cm]{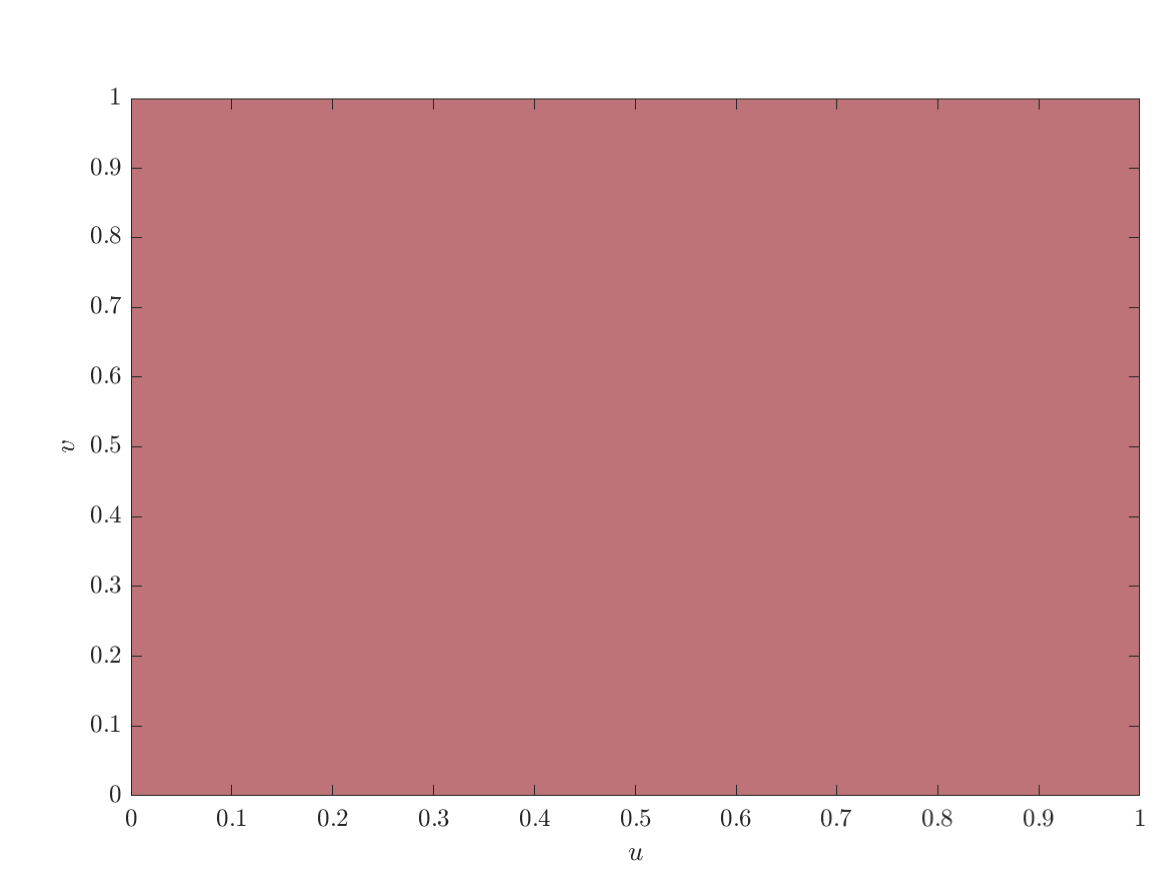}
		\caption{Graphon} 
		\label{er_graphon}
	\end{subfigure}
	\hspace{0.2cm}
	\begin{subfigure}{0.23\textwidth}
		\centering
		\includegraphics[height=2.8cm]{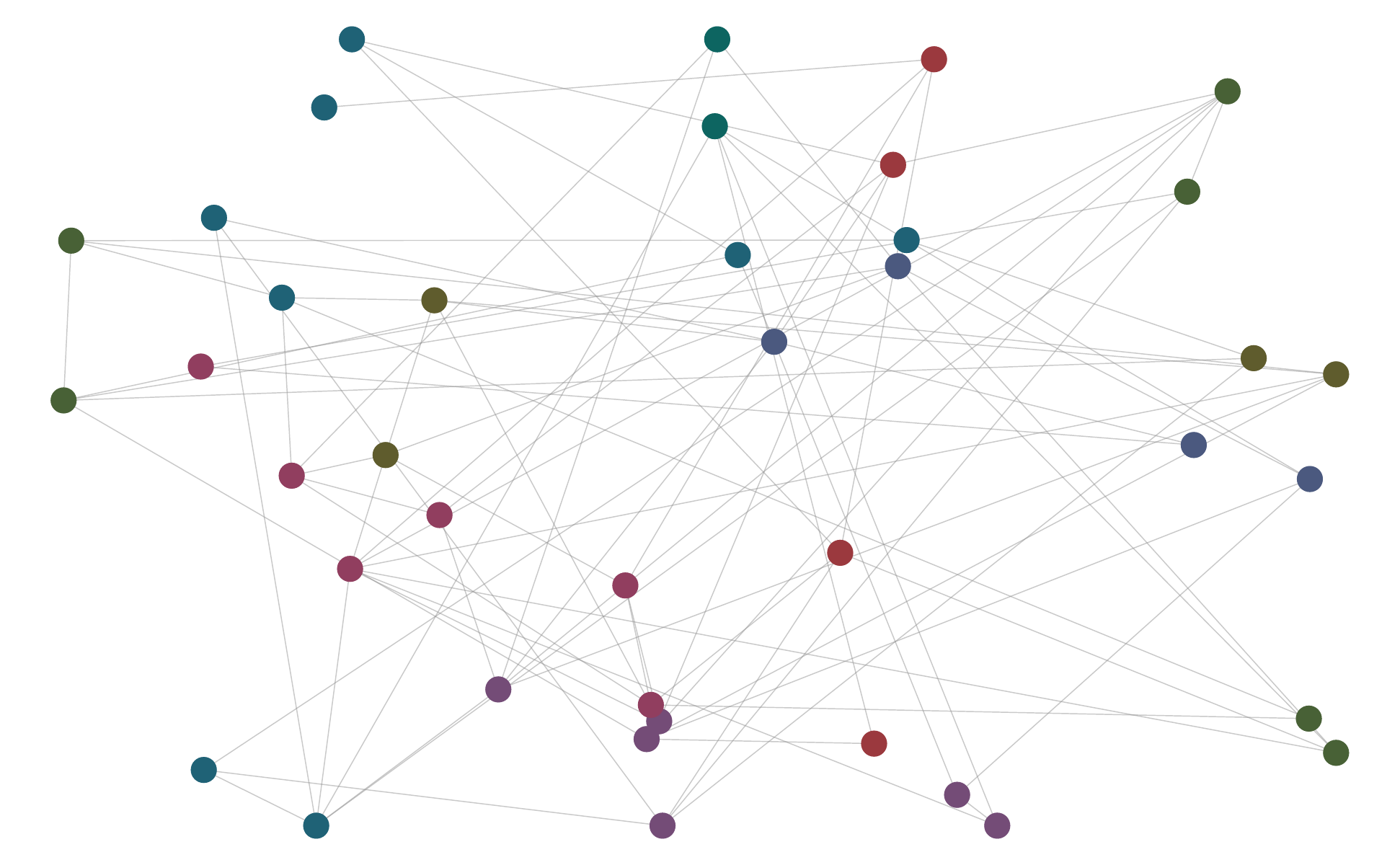}
		\caption{$n=20$} 
		\label{small_graph1}
	\end{subfigure}
	\hspace{0.2cm}
		\begin{subfigure}{0.23\textwidth}
		\centering
		\includegraphics[height=2.8cm]{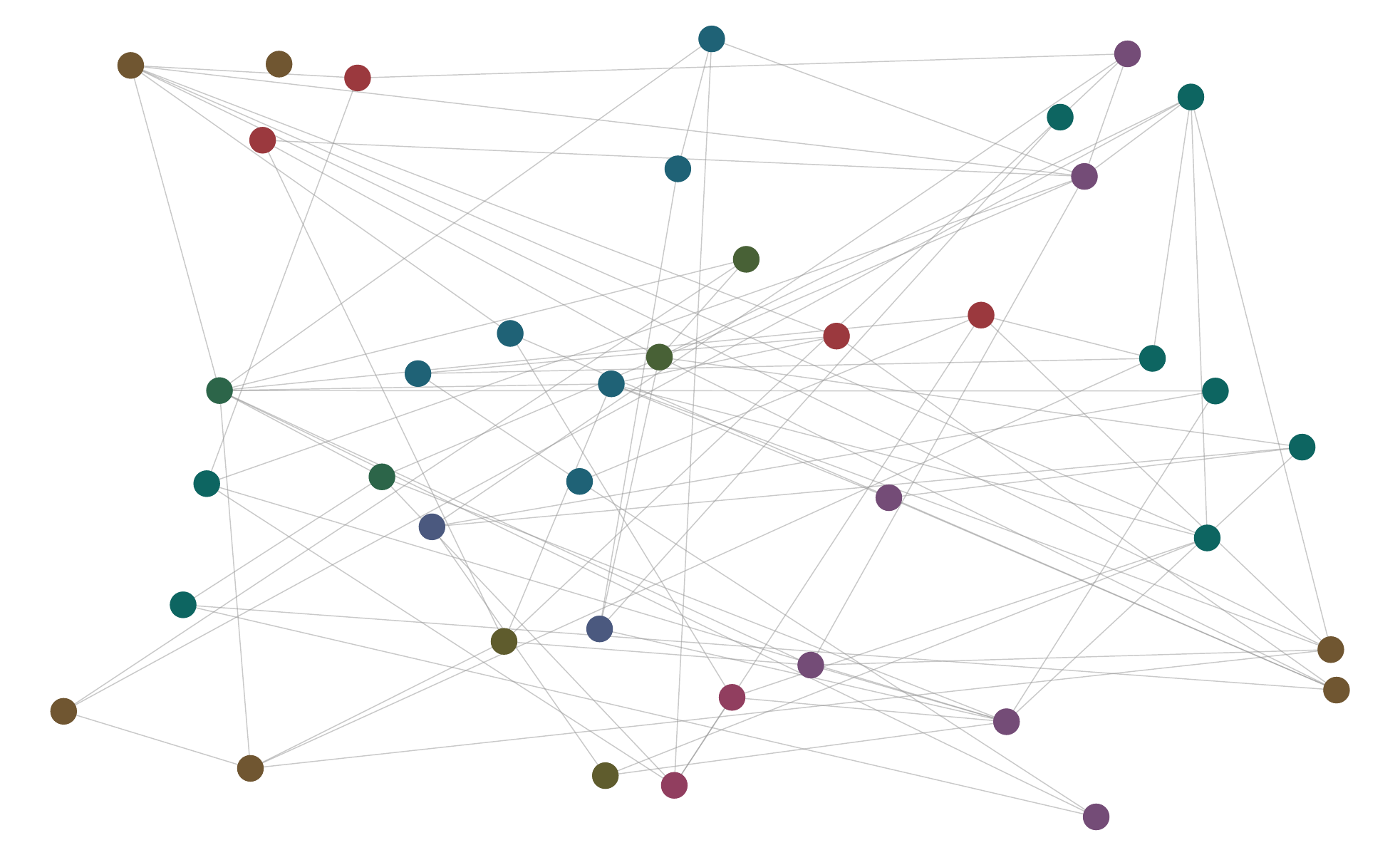}
		\caption{$n=20$} 
		\label{small_graph2}
	\end{subfigure}
	\hspace{0.2cm}
	\begin{subfigure}{0.23\textwidth}
		\centering
		\includegraphics[height=2.8cm]{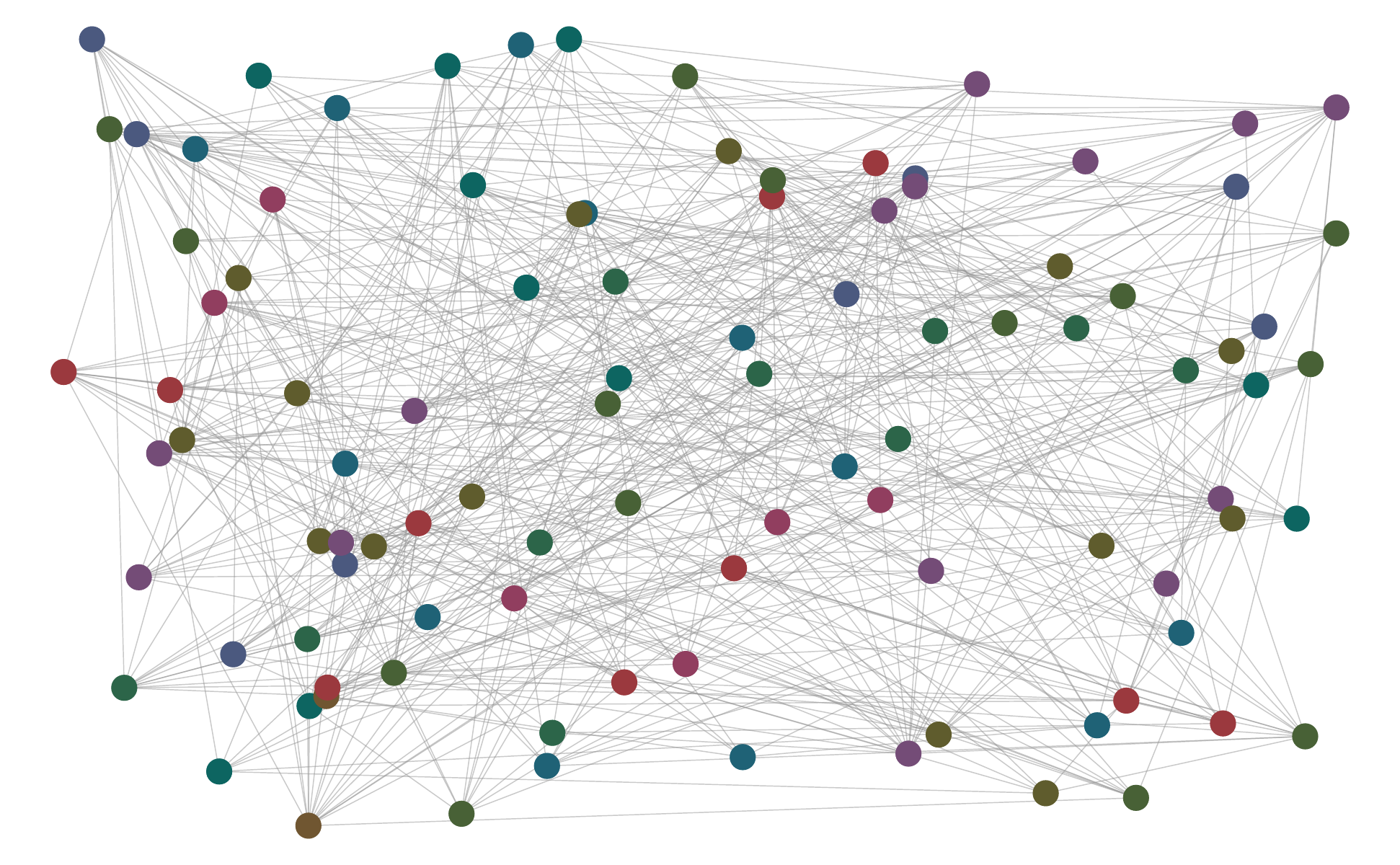}
		\caption{$n=50$} 
		\label{large_graph}
	\end{subfigure}
	\caption{Erd\"os–R\'enyi (constant) graphon with probability $p=0.2$ and three $n$-node graphs sampled from it, illustrating scenarios where WSP can be useful. We can analyze signals and design systems on a graphon, to apply them on a graph sampled from it (\subref{er_graphon} $\rightarrow$ \subref{small_graph1}); on a graph, to apply on another graph of same size (\subref{small_graph1} $\rightarrow$ \subref{small_graph2}); and on a small graph, to apply on a larger graph (\subref{small_graph2} $\rightarrow$ \subref{large_graph}).}
	\label{fig:experiments}
\end{figure*}

\blue{
Thms. \ref{thm:wft}--\ref{thm:non_disc2} are especially important because they provide theoretical justification for transferring signal analysis methods and information processing architectures across graphs arising from (or leading to) the same graphon. Indeed, the ability of GSP systems to be transferred between graphs, also known as \textit{transferability}, has been demonstrated empirically in network problems in wireless \cite{eisen2019optimal} and robotics \cite{tolstaya2020learning}. 
We identify three transferability scenarios for which the results in this work provide a theoretical foundation: \textbf{(S1)} graphon to graph (Figs. \ref{er_graphon}$\to$\ref{small_graph1}); \textbf{(S2)} graph to graph of same size (Figs. \ref{small_graph1}$\to$\ref{small_graph2}); and  \textbf{(S3)} graph to larger graph (Figs. \ref{small_graph2}$\to$\ref{large_graph}).

Attesting to the practical value of the WSP framework, each of these scenarios is illustrated in a numerical experiment in Sec. \ref{sec:sims}.
For instance, to showcase \textbf{(S1)}, we compare filter responses \textit{on a graphon and on a graph sampled from this graphon}.
The filter is a simple diffusion filter applied to a Gaussian Markov random field (GMRF). 
Interestingly, this example makes for a parallel with classical signal processing, where even if the application is digital, it is sometimes convenient to design and study filters in continuous time.
For \textbf{(S2)}, we perform signal analysis in \textit{different graphs of same size drawn from a common graphon}. 
Two $n$-node air pollution sensor networks are considered. As $n$ grows, we compare the Fourier transforms of the same air pollution signal on top of them. This illustrates the behavior of SP tools when applied to networks for which we only have access to an approximated or perturbed version of the graph.
Finally, we illustrate \textbf{(S3)} by \textit{transferring filters designed on small graphs to large graphs} in a movie recommendation example. 
Using real data from the MovieLens dataset, we calculate the optimal coefficients of a rating prediction filter on networks containing only a subset of all users, and then use it to predict movie ratings on the full user network. The goal of this experiment is to show that graph filters are transferable at scale, which significantly simplifies signal processing on large graphs.

The rest of this paper is organized as follows. Preliminary definitions are recalled in Sec. \ref{sec:graphs_graphons}. Sec. \ref{sec:gft} introduces the WSP framework and convergence results are stated in Sec. \ref{sec:filters}. Sec. \ref{sec:sims} details the numerical experiments described above (\textbf{S1}--\textbf{S3}). Proofs are deferred to the appendices.
Unless otherwise specified, $\|\cdot\|$ refers to the $L^2$ norm. When referring to the operator norm induced by the $L^2$ norm (spectral norm), we use the notation $\vertiii{\cdot}$.

%% file: graphs_and_graphons.tex

%
\begin{figure*}[t]
	\centering
	\begin{subfigure}{0.3\textwidth}
		\centering
		\includegraphics[height=3.5cm]{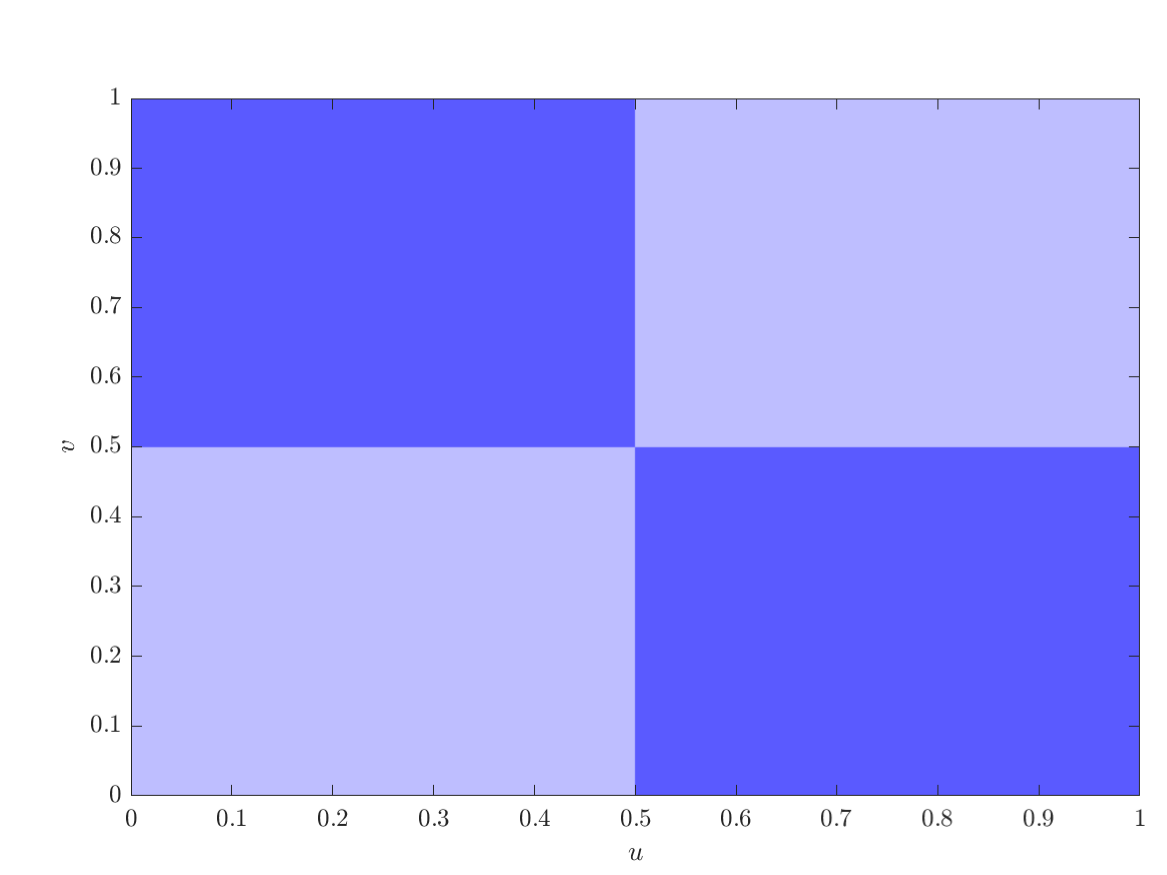}
		\centering
		\includegraphics[height=3.2cm]{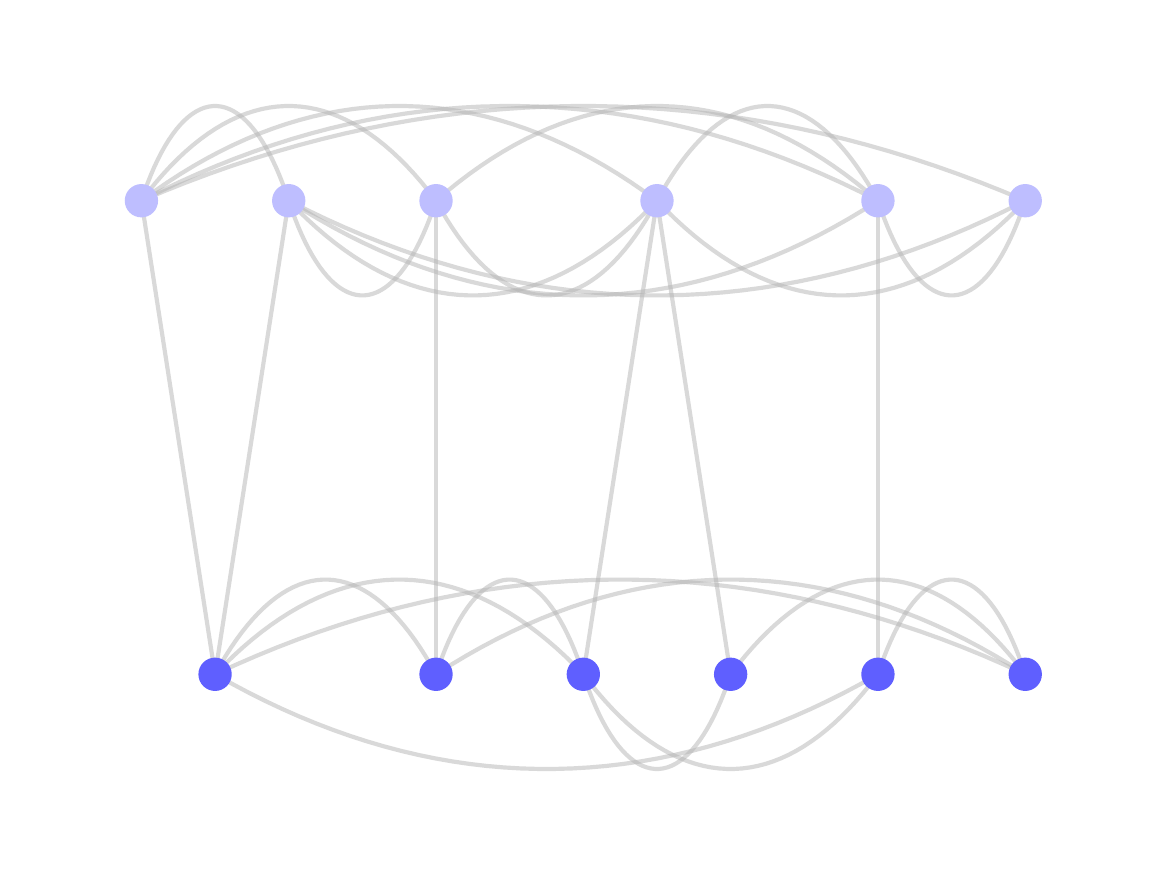}
		\caption{SBM with balanced communities} 
		\label{sbm1}
	\end{subfigure}
	\hspace{0.2cm}
	\begin{subfigure}{0.3\textwidth}
		\centering
		\includegraphics[height=3.5cm]{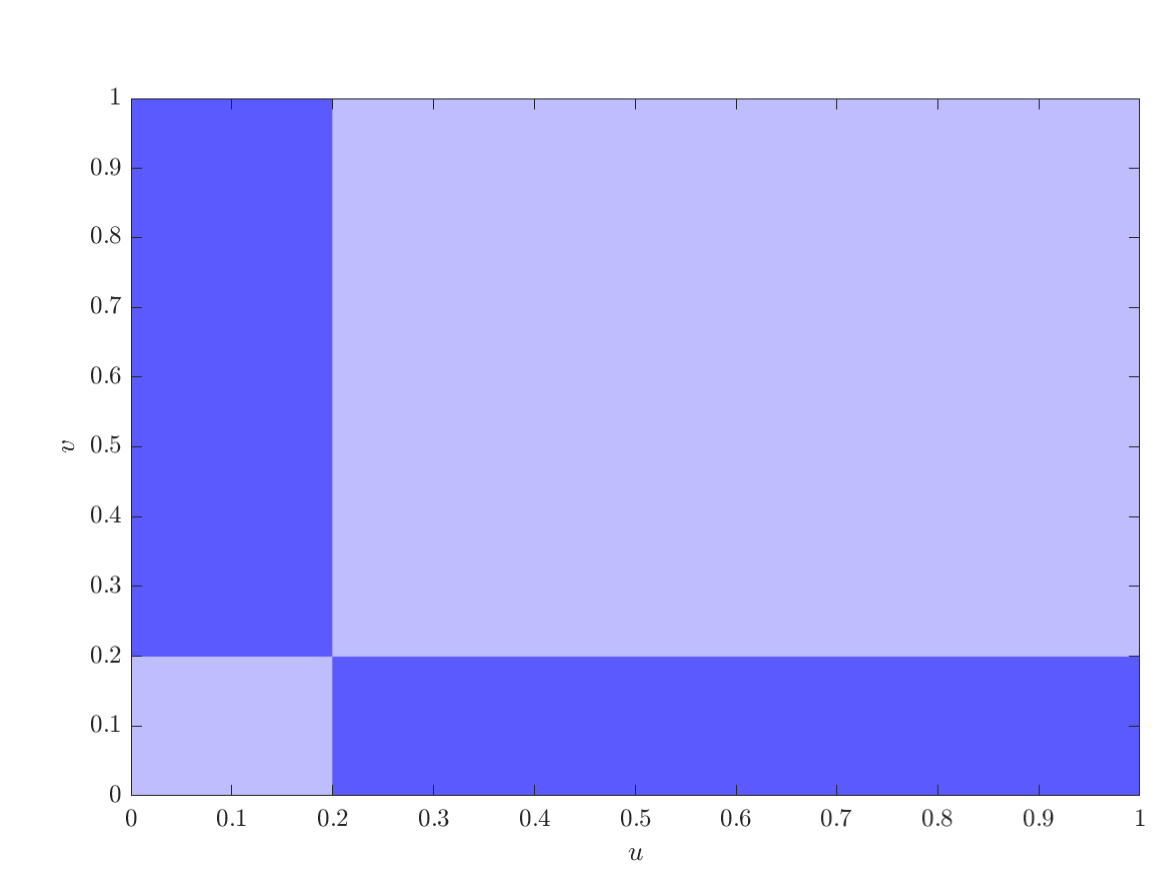}
		\centering
		\includegraphics[height=3.2cm]{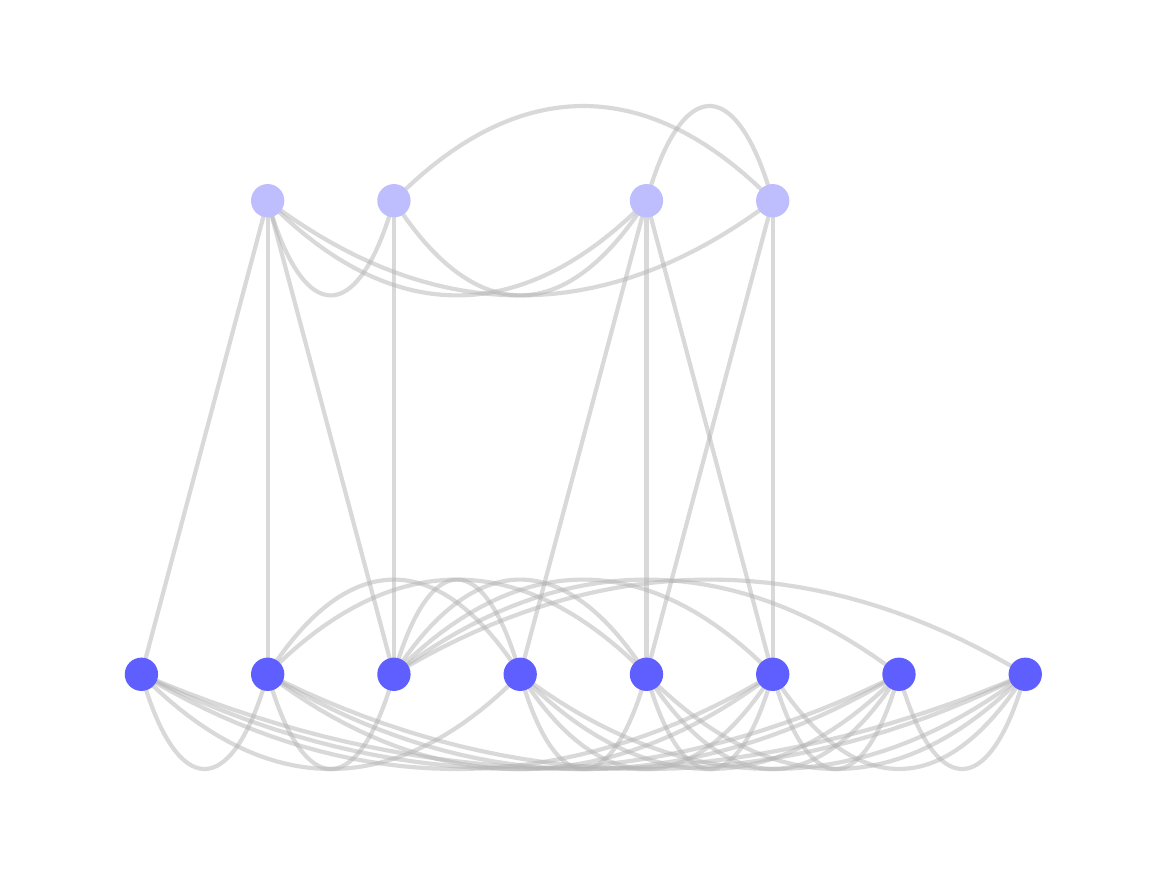} 
		\caption{SBM with imbalanced communities} 
		\label{sbm2}
	\end{subfigure}\hspace{0.2cm}
	\begin{subfigure}{0.3\textwidth}
		\centering
		\includegraphics[height=3.5cm]{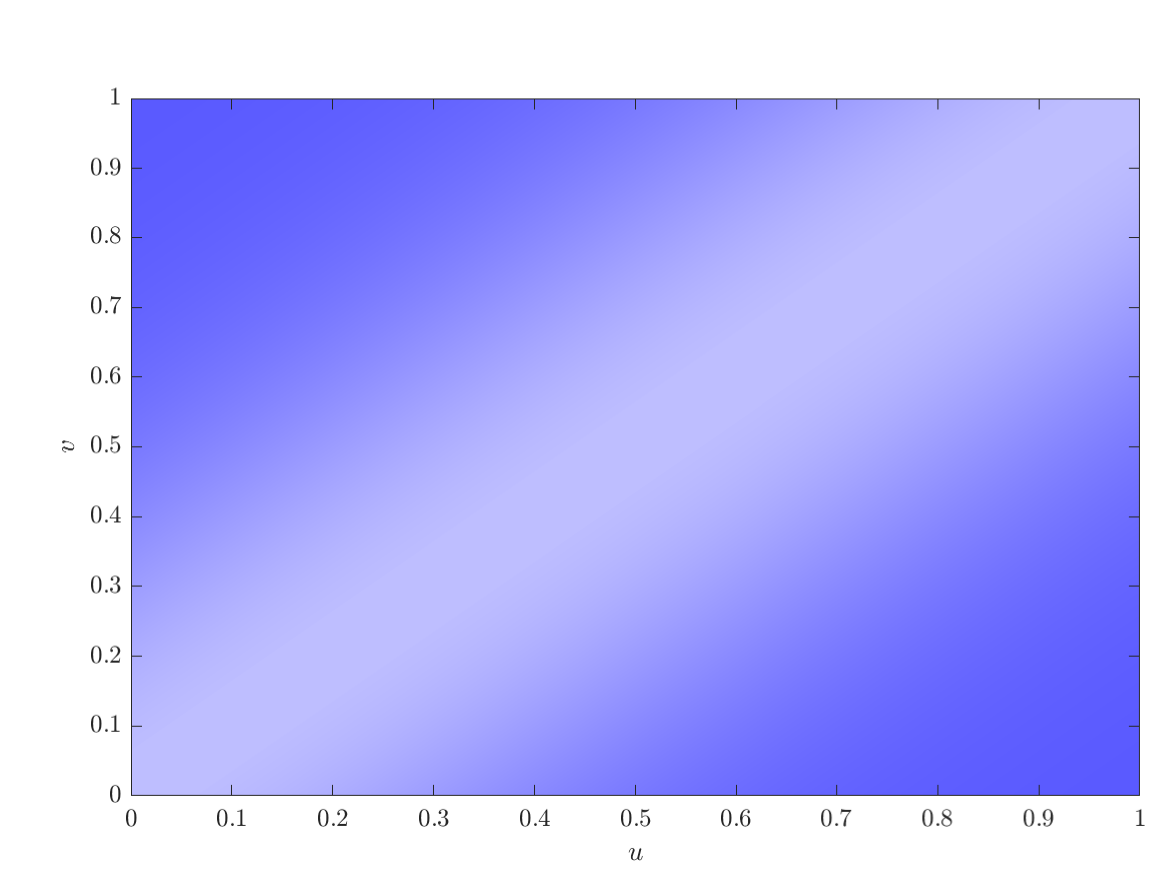}
		\centering
		\includegraphics[height=3.2cm]{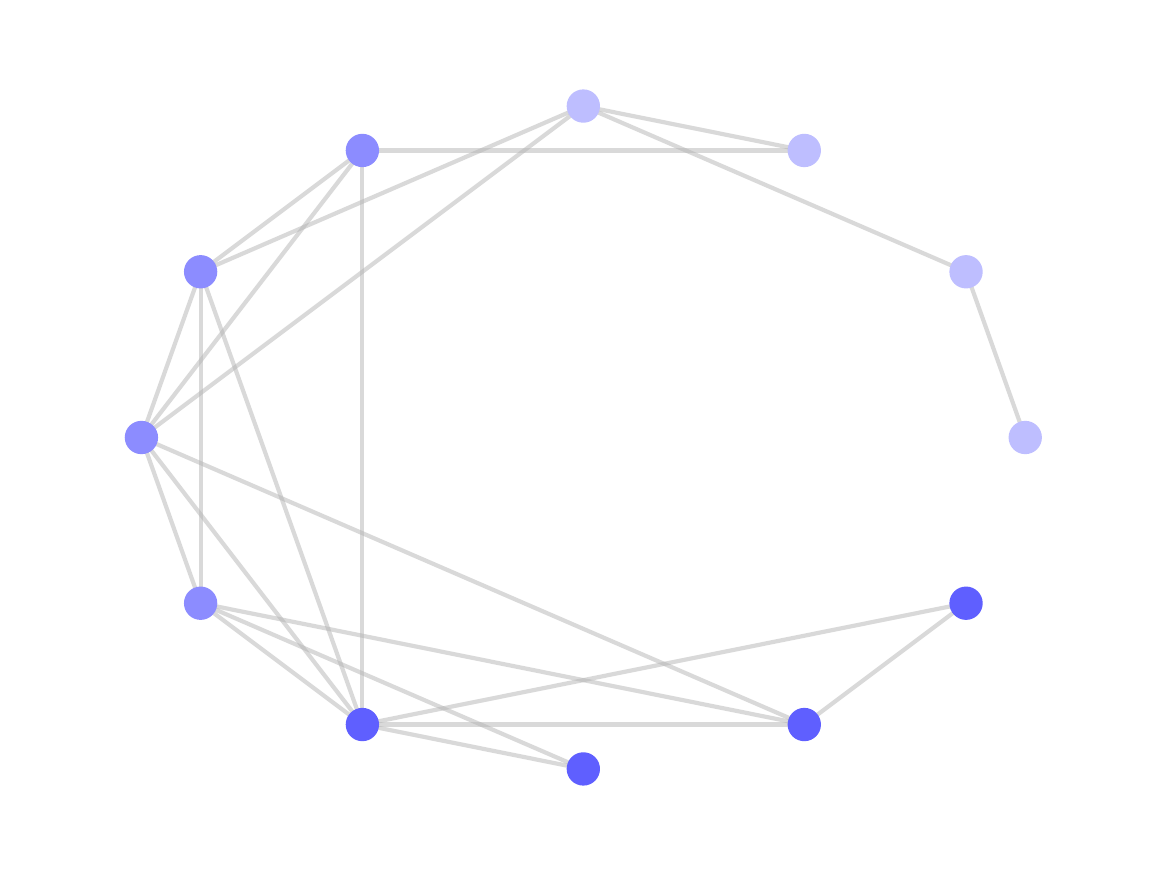} 
		\caption{Exponential} 
		\label{exp}
	\end{subfigure}
	\caption{Graphons and $12$-node $\bbW$-random graphs. Figs. \subref{sbm1} and \subref{sbm2} show SBM graphons and graphs with 2 communities and $p_{c_i c_i} = 0.8, p_{c_i c_j} = 0.2$. Fig. \subref{exp} shows an exponential graphon $\bbW(u,v) = \exp (-2.3(u-v)^2)$ and the corresponding graph.}
	\label{fig:graphon_examples}
\end{figure*}	

%
\section{Graphs and Graphons} \label{sec:graphs_graphons} 

Graphs are triplets $\bbG = (\ccalV, \ccalE, \ccalW)$ where $\ccalV$ is a set of $n$ nodes, $\ccalE \subseteq \ccalV \times \ccalV$ is a set of edges  and $\ccalW: \ccalE \to \reals$ is a weight function assigning weights $\ccalW(i,j) = w_{ij}$ to edges $(i,j)$ in $\ccalE$. The graph $\bbG$ can be equivalently represented by a number of matrix representations, which in the context of graph signal processing (GSP) are generically termed graph shift operators (GSOs). In this paper, we fix the GSO to be the \textit{adjacency matrix} $\bbS$
satisfying $S_{ij} = w_{ij}$ if and only if $(i,j)\in\ccalE$. We restrict our attention to undirected graphs with edge weights in the $[0,1]$ interval, so that $\bbS=\bbS^\Tr$ and $\bbS \in [0,1]^{n \times n}$. We will also use the notations $\bbG = (\ccalV, \ccalE, \ccalW)$ and $\bbG = (\ccalV, \ccalE, \bbS)$ interchangeably.

A graphon is a bounded symmetric measurable function 
\begin{align}\label{eqn_graphon}
   \bbW: [0,1]^2 \to [0,1]
\end{align}
which represents a graph with an uncountable number of nodes \cite[Chapter 7]{lovasz2012large},\cite{wolfe2013nonparametric}. 
By construction, graphons can also be interpreted as \blue{generative models for the edges of weighted or stochastic graphs}. Namely, by associating sample points $u_i \in [0,1]$ to nodes $i \in \{1,2,\ldots,n\}$, we can construct $n$-node graphs $\bbG_n$ where the edges are defined either by assigning edge weight $\bbW(u_i,u_j)$ to $(i,j)$ or by connecting $i$ and $j$ with probability 
\begin{align}\label{eqn_sampled_graph}
   p_{ij} = \bbW(u_i,u_j).
\end{align}
In the latter case, the $\bbG_n$ are unweighted. If, additionally, the $u_i$ are sampled independently and uniformly at random, these stochastic graphs are called \textit{$\bbW$-random graphs}. Three examples of graphons and of $\bbW$-random graphs sampled from them are shown in Fig. \ref{fig:graphon_examples}. The one in Fig. \ref{sbm1} is a stochastic block model (SBM) graphon with two balanced communities where the intra-community edge probability is $0.8$ and the inter-community edge probability is $0.2$. The one in Fig. \ref{sbm2} is also a SBM graphon with the same inter- and intra-community probabilities, but with unbalanced communities. The one in Fig. \ref{exp} is an exponential graphon, which can be used to generate graphs where nodes are connected if their labels $u_i$ and $u_j$ are close.

\subsection{Convergent sequences of graphs}

A second and perhaps more interesting interpretation of graphons is as the limit objects of convergent graph sequences. A sequence of graphs $\{\bbG_n\}$ is said to converge if and only if the density of homomorphisms between any finite, undirected and unweighted graph $\bbF = (\ccalV', \ccalE')$, which we call a \textit{motif}, and the $\bbG_n$ converges \cite{lovasz2006limits}. Homomorphisms between a motif $\bbF$ and an arbitrary graph $\bbG = (\ccalV, \ccalE, \bbS)$ are adjacency preserving maps from $\ccalV'$ to $\ccalV$, i.e., a map $\beta: \ccalV' \to \ccalV$ is a homomorphism if, for every $(i,j) \in \ccalE'$, $(\beta(i),\beta(j)) \in \ccalE$. The graph $\bbF$ can thus be interpreted as a graph pattern that we want to ``identify'' in $\bbG$. 

A motif $\bbF$ will typically occur in multiple locations of the graph $\bbG$. Thus, we can count the number of homomorphisms between $\bbF$ and $\bbG$, which we denote $\mbox{hom}(\bbF, \bbG)$. Since there are a total of $|\ccalV|^{|\ccalV'|}$ possible maps between the vertices of $\bbF$ and $\bbG$ but only a fraction of them are homomorphisms, we further define the \textit{density of homomorphisms} from $\bbF$ to $\bbG$ as
\begin{equation} \label{eqn:hom_graph}
t(\bbF,\bbG) = \dfrac{\mbox{hom}(\bbF,\bbG)}{|\ccalV|^{|\ccalV'|}} = \dfrac{\sum_\beta \prod_{(i,j) \in \ccalE'} [\bbS]_{\beta(i) \beta(j)}}{|\ccalV|^{|\ccalV'|}}\ .
\end{equation}
This is easiest to understand when $\bbG$ is unweighted, in which case $t(\bbF,\bbG)$ is simply the fraction of the total number of ways in which the motif $\bbF$ can be mapped into $\bbG$.

The concept of homomorphism densities can also be generalized to graphons. We define the density of homomorphisms between the motif $\bbF$ and the graphon $\bbW$ as 
\begin{equation} \label{eqn:hom_graphon}
t(\bbF,\bbW) = \int_{[0,1]^{|\ccalV'|}} \prod_{(i,j) \in \ccalE'} \bbW(u_i,u_j) \prod_{i \in \ccalV'} du_i\ .
\end{equation}
This can be interpreted as the probability of sampling the motif $\bbF$ from the graphon $\bbW$. With these definitions in hand, a sequence of undirected graphs $\{\bbG_n\}$ is said to converge to the graphon $\bbW$ if, for all finite simple graphs $\bbF$,
\begin{equation}\label{eqn_convergence}
   \lim_{n \to \infty} t(\bbF,\bbG_n) = t(\bbF, \bbW)\ .
\end{equation}
In this case, we refer to $\bbW$ as the \textit{limit graphon} of the sequence. This form of convergence is called ``convergence in the homomorphism density sense''. An example of convergent graph sequence that is easy to visualize is that of a sequence of $\bbW$-random graphs. The sequence of graphs $\{\bbG_n\}$ generated by sampling $\{u_i\}_{i=1}^n$ uniformly at random as $n \to \infty$ can be shown to converge in the homomorphism density sense with probability one \blue{\cite[Example 11.6, Lemma 11.8]{lovasz2012large}}. 

To conclude, we point out that, while the two interpretations of a graphon ---  as a generative model for graph families and as the limit object of graph sequences --- are theoretical, their practical value lies in that they can be used to identify sets of graphs with large number of nodes and similar structure. This simplifies the study of the properties of large graphs. 

\subsection{Convergence in cut norm}

Similarly to how graphs can be obtained by sampling or evaluating a graphon, graphons can be defined, or \textit{induced}, by graphs. 
Every undirected graph $\bbG = (\ccalV, \ccalE, \bbS)$ with $|\ccalV| = n$ and $\bbS \in [0,1]^{n \times n}$
admits an induced graphon representation $\bbW_\bbG$. This graphon is obtained in two steps. First, we construct a \blue{regular} partition $I_1 \cup \ldots \cup I_n$ of $[0,1]$\blue{, i.e., the partition given by $I_j = [(j-1)/n,j/n)$ for $1 \leq j \leq n-1 $ and $I_n = [(n-1)/n,1]$. Then, the induced graphon $\bbW_\bbG$ is defined as \cite[Chapter 7.1]{lovasz2012large},\cite[Sec. 5]{avella2018centrality}}
\begin{equation}\label{eqn_induced_graphon}
   \bbW_\bbG(u,v) = \sum_{j=1}^n \sum_{k=1}^n [\bbS]_{jk} \times \mbI(u \in I_j)\mbI(v \in I_k)\ .
\end{equation}

The concept of induced graphon is useful to define a second mode of convergence for graph sequences --- convergence in cut norm. The cut norm of a graphon $\bbW$ is defined as \blue{\cite[eq. (8.13)]{lovasz2012large}}
\begin{equation} \label{eqn:cut_norm_def}
\|\bbW\|_\square = \sup_{S,T \subseteq [0,1]} \bigg| \int_{S \times T} \bbW(u,v)du dv \bigg| 
\end{equation}
i.e., it is equal to the size of its maximum cut. The following lemma, adapted from \cite[Theorem 11.57]{lovasz2012large}, states that if a sequence of graphs $\{\bbG_n\}$ converges to $\bbW$ in the homomorphism density sense, then it also converges to $\bbW$ in the cut norm.
\begin{lemma}[Cut norm convergence] \label{cut_norm_conv}
If $\{\bbG_n\}\to\bbW$ in the homomorphism density sense, then there exists a sequence of permutations $\{\pi_n\}$ such that
\begin{equation} \label{eqn:cut_norm_conv}
\|\bbW_{\pi_n(\bbG_n)} - \bbW\|_\square \to 0
\end{equation}
where $\bbW_{\pi_n(\bbG_n)}$ is the graphon induced by the graph $\pi_n(\bbG_n)$.
\end{lemma}

Therefore, for every convergent sequence $\{\bbG_n\}$ there exists a non-empty set of permutation sequences $\{\pi_n\}$ for which the cut norm of the induced graphons $\bbW_{\pi_n(\bbG_n)}$ converges as in \eqref{eqn:cut_norm_conv}. This is formalized in Def. \ref{def:permutation_sequence}.

\begin{definition}[Set of admissible permutations]\label{def:permutation_sequence}
Given a sequence $\{\bbG_n\}$ converging to $\bbW$ in the homomorphism density sense, the family of convergent permutation sequences $\ccalP$ associated with $\{\bbG_n\}$ is defined as
\begin{equation*}
\ccalP = \bigg\{\{\pi_n\} \ |\ \|\bbW_{\pi_n(\bbG_n)} - \bbW\|_\square \to 0\bigg\}\ .
\end{equation*}
\end{definition}

The set $\ccalP$ will be especially important in the convergence analyses of Secs. \ref{sec:gft}--\ref{sec:filters}. In particular, in the definition of convergent sequences of graph signals (Def. \ref{defn:graph_signal_conv}), we will use permutation sequences $\{\pi_n\} \in \ccalP$ to ``organize'' the signals on the graphs of a convergent sequence so that the labeling of the signals matches the node labeling for which the graphs converge. 


\subsection{Graph Signal Processing} \label{sbs:gsp}

GSP deals with signals defined on top of a graph $\bbG$. Formally, a graph signal is a map from the vertex set $\ccalV$ onto the real numbers, which we write as the pair $(\bbG, \bbx)$ and where the $i$th component $x_i$ is the value of the signal at node $i$. The three fundamental concepts of GSP are shift operations, LSI filters and graph Fourier transforms (GFTs). We say that $\bbz$ is the result of shifting $\bbx$ on the graph $\bbS$ if $\bbz=\bbS\bbx$. Shifts can be composed to produce $k$-order shifted signals $\bbS^k\bbx$ and, as in the case of time signals, a weighted sum of shifted signals defines the LSI filter $\bbH(\bbS)$ as the linear map
\begin{equation} \label{eqn:lsi-gf}
   \bby = \bbH(\bbS) \bbx = \sum_{k=0}^{K} h_k \bbS^k \bbx\ .
\end{equation}
In \eqref{eqn:lsi-gf}, the weights $h_k$ are called graph filter taps \cite{segarra17-linear}. The filter $\bbH(\bbS)$ is said to be shift-invariant because, if $\bby = \bbH(\bbS)\bbx$ and we shift the input to $\bbx' = \bbS\bbx$, the output $\bby' = \bbH(\bbS)\bbx'$ is simply the shifted version of $\bby$, $\bby' = \bbS\bby$. The LSI filter $\bbH(\bbS)$ is also sometimes referred to as a graph convolutional filter \cite{ruiz2020gnns}. 

Since $\bbS$ is symmetric, it is diagonalizable as $\bbS = \bbV \bbLam \bbV^{\Hr}$. The matrix $\bbLam$ is the diagonal matrix of eigenvalues and the columns of $\bbV$ are the graph eigenvectors. \blue{Herein, we assume that the eigenvalues are ordered according to their sign and in decreasing order of absolute value, i.e., $ \lambda_1(\bbS) \geq \lambda_2(\bbS) \geq \ldots \geq 0 \geq \ldots \geq \lambda_{-2}(\bbS) \geq \lambda_{-1}(\bbS)$ (e.g., Fig. \ref{fig:eigenvalues_graphon}).} The graph Fourier transform (GFT) of the graph signal $(\bbG,\bbx)$ is then defined as the projection of $\bbx$ onto the eigenvector basis $\bbV$, i.e.,
\begin{equation} \label{eqn:gft}
   \hbx = \mbox{GFT} \{(\bbG, \bbx)\} = \bbV^\Hr \bbx\ .
\end{equation}
In particular, $[\hbx]_j = \bbv_j^\Hr \bbx$ for each index $j$, where $\bbv_j$ is the eigenvector associated with $\lambda_j$. This operation has the effect of decomposing $(\bbG, \bbx)$ in the eigenbasis of the graph, which makes sense if we interpret the eigenvalues as frequencies. Similarly, the inverse graph Fourier transform (iGFT) is defined as
\begin{equation*}
   \mbox{iGFT} \{\hbx\} = \bbV \hbx = \bbx\ .
\end{equation*}
Since $\bbV^\Hr\bbV=\bbI$, the iGFT is a proper inverse and can recover $\bbx$ from $\hbx$ without loss of information.

The LSI filters defined in \eqref{eqn:lsi-gf} also admit a spectral representation $\hbH(\bbLam)$. This spectral representation is given by
\begin{equation} \label{eqn:lsi_graph_filter_resp}
   \hbH(\bbLam)= \sum_{k=0}^{K} h_k \bbLam^k\ .
\end{equation}
Therefore, if we consider the action of the filter $\bbH(\bbS)$ in the frequency domain, we see that $\hby = \hbH(\bbLam)\hbx$, i.e., graph filters are pointwise operators in the GFT domain. 

Another interesting observation regarding the spectral response of graph filters is that, for any set of filter taps, we can define the frequency response \cite{elvin19-spectral} 
\begin{equation} \label{eqn:any_graph_filter_resp}
   h(\lam)= \sum_{k=0}^{K} h_k \lam^k \ .
\end{equation}
Comparing \eqref{eqn:lsi_graph_filter_resp} and \eqref{eqn:any_graph_filter_resp}, we see that using the same set of coefficients on different graphs induces different responses depending on the eigenvalues of $\bbS$. Indeed, if we let $\lam_i$ denote the $i$th eigenvalue of $\bbS$ and $\hhatx_i$ and $\hhaty_i$ the $i$th components of the GFTs $\hbx$ and $\hby$, we have that $\hhaty_i = h(\lam_i)\hhatx_i$.
 
The goal of this paper is to generalize the definitions of graph signals, GFTs, and convolutional graph filters to graphon signals, graphon Fourier transforms, and convolutional graphon filters (Sec. \ref{sec:gft}). In introducing this graphon signal processing framework, we intend to do for large-scale graph signal processing what graphons do for very large and dynamic networks. Namely, we leverage the fact that working with limits is easier and focus on graph signal limits --- graphon signals --- to facilitate GSP on large-scale and dynamic networks. 
Importantly, we also show that, for sequences of graphs that converge to a graphon in the sense of \eqref{eqn_convergence}, corresponding sequences of GFTs and graph filters converge to the respective graphon Fourier transforms and graphon filters (Sec. \ref{sec:filters}).

%% file: gft-graphon.tex





The central concept in graphon signal processing is the graphon signal. Graphon signals are defined as pairs $(\bbW, X)$ where the function $X: [0,1] \to \mbR$ maps points of the unit interval to the real numbers, i.e., the signal values. The graphon signals that we consider have finite energy, i.e., $X$ is a function in $L^2([0,1])$. As in the case of graphons, graphon signals can be \textit{induced by graph signals}. Given a $n$-node graph signal $(\bbG,\bbx)$, the induced graphon signal $(\bbW_\bbG, X_\bbG)$ is defined as
\begin{equation} \label{eqn:graphon_signals_induced}
X_\bbG(v) = \sum_{j=1}^n [\bbx]_j \times \mbI(v \in I_j)
\end{equation}
where $\bbW_\bbG$ is the graphon induced by $\bbG$ [cf. \eqref{eqn_induced_graphon}] and $I_1 \cup \ldots \cup I_n$ is the \blue{regular} partition of the unit interval, i.e., $I_j = [(j-1)/n,j/n)$ for $1 \leq j \leq n-1 $ and $I_n=[(n-1)/n,1]$.

\subsection{Convergent sequences of graph signals}

We define convergent sequences of graph signals as follows.

\begin{definition}[Convergent sequences of graph signals] \label{defn:graph_signal_conv}
A sequence of graph signals $\{(\bbG_n,\bbx_n)\}$ is said to converge to the graphon signal $(\bbW,X)$ if, for every graph motif $\bbF$,
\begin{equation*}
\lim_{n \to \infty} t(\bbF,\bbG_n) \to t(\bbF,\bbW)
\end{equation*}
and if there exists a sequence of permutations $\{\pi_n\} \in \ccalP$ such that
\begin{equation*}
\lim_{n \to \infty} \|X_{\pi_n(\bbG_n)} - X\| = 0
\end{equation*}
where $\ccalP$ is the set of admissible permutations for the sequence $\{\bbG_n\}$ [cf. Def. \ref{def:permutation_sequence}] and $(\bbW_{\pi_n(\bbG_n)},X_{\pi_n(\bbG_n)})$ is the graphon signal induced by the permuted graph signal $(\pi_n(\bbG_n), \pi_n(\bbx_n))$ [cf. \eqref{eqn:graphon_signals_induced}].
\end{definition}

A sequence of graph signals is thus convergent if (i) the underlying graphs converge and (ii) the graphon signals induced by some permutation sequence $\{\pi_n\} \in \ccalP$ of the graph signals converge in $L^2$. The role of the permutation sequence is to match the labels of the signals $\bbx_n$ to those of the sequence of graphs $\bbG_n$ that converges in cut norm [cf. Lemma \ref{cut_norm_conv}]. 
This is a similar requirement to the isometric embeddings necessary to define Gromov-Hausdorff convergence in metric spaces \cite{kerr2009gromov}.

\blue{Importantly, the graph signal limit~$(\bbW,X)$ is unique for each $\{\pi_n\} \in \ccalP$. Indeed, suppose that it is not, i.e., that~$\|X_{\pi_n(\bbG_n)}-X\|  \to 0$ and~$\|X_{\pi_n(\bbG_n)}-Y\| \to 0$ with~$\|X-Y\| \geq \epsilon > 0$. Using the triangle inequality, we get
\begin{align*}
\begin{split}
\|X-Y\|  &= \|X-X_{\pi_n(\bbG_n)}+X_{\pi_n(\bbG_n)}-Y\|  \\
&\leq \|X-X_{\pi_n(\bbG_n)}\|  + \|X_{\pi_n(\bbG_n)}-Y\|  \to 0
\text{,}
\end{split}
\end{align*}
which contradicts the hypothesis since there must then exist~$n_0$ such that $\|X-X_{\pi_n(\bbG_n)}\|  + \|X_{\pi_n(\bbG_n)}-Y\| <\epsilon$ for $n>n_0$.} See also Remark \ref{remark:limit_uniqueness}.

\subsection{The graphon operator and graphon filters}

Every graphon $\bbW$ induces an integral operator $T_{\bbW}: L^2([0,1]) \to L^2([0,1])$, which maps a signal $(\bbW, X)$ to the signal $(\bbW, Y)$ given by
\begin{equation} \label{eqn:graphon_shift}
Y(v) = (T_\bbW X)(v) = \int_0^1 \bbW(u,v)X(u)du\ .
\end{equation}
We refer to $T_\bbW$ as the \textit{graphon shift operator} (WSO) because it induces a diffusion of $(\bbW, X)$ on the graphon analogous to the diffusion induced by the adjacency matrix $\bbS$ on a graph. Building upon this parallel, LSI graphon filters are defined as follows.

%
\begin{figure}[t]
    \centering
    \input{plots_stability/eigenvalues_graphon2.tex}    
    \caption{Graphon eigenvalues. A graphon has an infinite number of eigenvalues $\lambda_i$ but for any fixed constant $c$ the number of eigenvalues $|\lambda_i|\geq c$ is finite. Thus, eigenvalues accumulate at $0$ and this is the only accumulation point for graphon eigenvalues. The quantity $d_n$ is approximately equal to the minimum distance between $c$ (or $-c$) and the eigenvalues in the set $\{\lambda_i\ |\ \lambda_i < c\}$.}
    \label{fig:eigenvalues_graphon}
\end{figure}
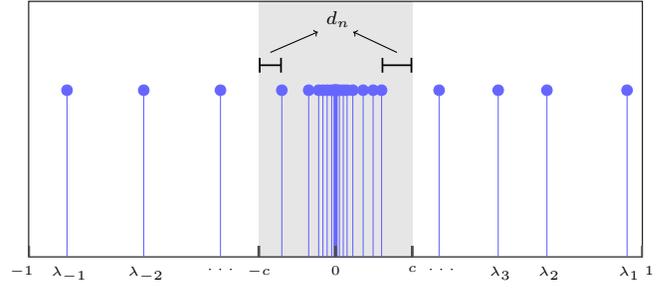

\begin{definition}[LSI graphon filters] \label{defn:linear}
Let $(\bbW,X)$ be a graphon signal. A LSI graphon filter $T_\bbH: L^2([0,1]) \to L^2([0,1])$ maps $(\bbW, X) \mapsto (\bbW, Y)$ with $ (\bbW, Y)$ given by
\begin{align} \label{eqn:lsi-wf}
\begin{split}
Y(v) = (T_\bbH X)(v) &= \sum_{k=0}^K h_k (T_{\bbW}^{(k)} X)(v) \\
\mbox{where } (T_{\bbW}^{(k)}X)(v) &= \int_0^1 \bbW(u,v)(T_\bbW^{(k-1)} X)(u)du,\ k \geq 1
\end{split}
\end{align}
and $T_{\bbW}^{(0)} = \mathds{I}$, the identity operator. The $h_k$ are known as the \textit{graphon filter taps}. 
\end{definition}

Similarly to LSI graph filters, the graphon filters in Def. \ref{defn:linear} are shift-invariant because, given an input graphon signal $X$ and denoting the graphon filter ouput $Y=T_\bbH X$, applying $T_\bbH$ to a shifted version of the input $X'=T_\bbW X$ yields $Y' = T_\bbH X' = T_\bbW Y$, i.e., $Y'$ is a shifted version of $Y$.


\subsection{Graphon spectra and the graphon Fourier transform}

Because $\bbW$ is a bounded symmetric function, $T_\bbW$ is a self-adjoint Hilbert-Schmidt operator (see Lemma \ref{lemma:HS}, Appendix \ref{appx:HS}). As such, it can be decomposed in the operator's basis as 
\begin{equation} \label{eqn:graphon_spectra}
\bbW(u,v) = \sum_{i \in \mbZ \setminus \{0\}} \lambda_i \varphi_i(u) \varphi_i(v)
\end{equation}
with eigenvalues $\lambda_i \in [-1,1]$, and eigenfunctions $\varphi_i: [0,1] \to \mbR$. As before, we separate positive and negative eigenvalues by ordering them according to their sign and in decreasing order of absolute value. Therefore, we have $1 \geq \lambda_1 \geq \lambda_2 \geq \ldots \geq 0 \geq \ldots \geq \lambda_{-2} \geq \lambda_{-1} \geq -1$.
The eigenfunctions form an orthonormal basis of $L^2([0,1])$. Note that the eigenvalues, and hence the eigenfunctions, are countable. What is more, since $\|\bbW\|^2 \leq 1$ the trace of $T_\bbW^* T_\bbW$ is bounded by one and so the $\lambda_i$ converge to $0$ for $|i| \to \infty$ as depicted in Fig. \ref{fig:eigenvalues_graphon}. Zero is the only point of accumulation, which in turn implies that all $\lambda_i \neq 0$ have finite multiplicity \cite{lax02-functional}.

%
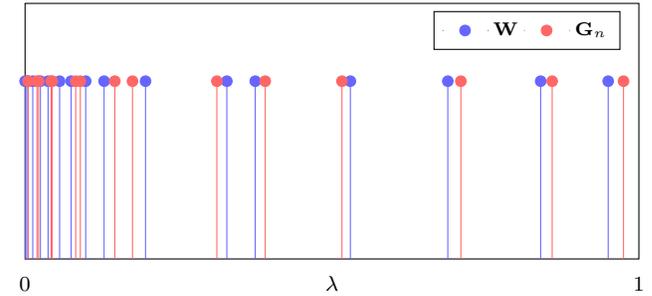
\begin{figure}[t]
    \centering
    \input{plots_stability/eigenvalues_graph_graphon.tex}    
    \caption{Comparison of graphon eigenvalues (blue) and eigenvalues of a graph $\bbG_n$ taken from a convergent graph sequence (red). Only the positive eigenvalues are depicted. For $n\to\infty$, the eigenvalues of $\bbG_n$ converge to the eigenvalues of $\bbW$. However, the accumulation of graphon eigenvalues close to $\lam=0$ means that the GFT converges to the WFT only for graphon bandlimited signals.}
    \label{fig:eigenvalues_graph_graphon}
\end{figure}

Eq. \eqref{eqn:graphon_spectra} allows writing $T_\bbW$ as
\begin{equation} \label{eqn:graphon_shift2}
(T_\bbW X)(v) = \sum_{i \in \mbZ \setminus \{0\}} \lambda_i \varphi_i(v) \int_0^1 \varphi_i(u)X(u)du \text{.}
\end{equation}
The integral terms $\int_0^1 \varphi_i(u)X(u)du$ are the $L^2$ inner products $\langle X, \varphi_i\rangle$ between the signal $X$ and the eigenfunctions $\varphi_i$. Since the $\varphi_i$ form a complete orthonormal basis of $L^2([0,1])$, the inner products $\langle X, \varphi_i\rangle$ provide a complete representation of $(\bbW, X)$ on the graphon basis. Although there is an infinite number of eigenfunctions, they are countable and so the change of basis can always be defined. This change of basis operation is called the \textit{graphon Fourier transform} (WFT).
\begin{definition}[Graphon Fourier Transform]  \label{defn:wft}
Consider the graphon signal $(\bbW,X)$, and let $\{\lambda_i\}_{i \in \mbZ \setminus \{0\}}$ and $\{\varphi_i\}_{i \in \mbZ \setminus \{0\}}$ be the eigenvalues and eigenfunctions of $T_\bbW$. 
The graphon Fourier transform (WFT) of $(\bbW,X)$ is defined as 
\begin{align*}
\begin{split}
&\mbox{WFT}\left[(\bbW,X)\right] = \hat{X} \quad \text{with}\\
&[\hat{X}]_i = \hat{X}(\lambda_i) = \int_{0}^1 X(u) \varphi_i(u) du\ .
\end{split}
\end{align*}
The inverse Graphon Fourier Transform (iWFT) of $\hat{X}$ is defined as
\begin{equation*}
\mbox{iWFT}\big[\hat{X}\big] = \sum_{i\in \mbZ \setminus \{0\}} \hat{X}(\lambda_i) \varphi_i = X\ .
\end{equation*}
\end{definition}

Since the $\{\varphi_i\}_{j \in \mbZ \setminus \{0\}}$ are orthonormal, the iWFT is the proper inverse transformation of the WFT. 
Def. \ref{defn:wft} further allows defining graphon signals that are \textit{bandlimited}.

\begin{definition}[Bandlimited graphon signals] \label{defn:bandlimited}
A graphon signal $(\bbW,X)$ is \textit{$c$-bandlimited} with bandwith $c \in [0,1]$ if $\hat{X}(\lambda_i) = 0$ for all $i$ such that $|\lambda_i| < c$.
\end{definition}
Because all nonzero eigenvalues have finite multiplicity, the WFT of a bandlimited graphon signal is finite-dimensional. Bandlimited graphon signals have two noteworthy properties. The first is that they extend the notion of graph bandlimited signals, which are the most common type of graph signal in practical GSP applications \cite{ramakrishna2020user}.
The second is that, since they only depend on a finite number of graphon eigenfunctions, their WFT can be computed analytically. Although countability of the $\varphi_i$ allows us to write the definition of the WFT (Def. \ref{defn:wft}) for any graphon signal, calculating all inner products $\langle X, \varphi_i\rangle$ is infeasible because the graphon basis is infinite-dimensional. 

\subsection{Graphon filter frequency response}

The WFT also allows obtaining the frequency response of graphon filters. Using the spectral decomposition of $T_\bbW$ \eqref{eqn:graphon_spectra}, we can rewrite the LSI graphon filter \eqref{eqn:lsi-wf} as
\begin{equation*}
Y(v) = (T_\bbH X)(v) = \sum_{i\in \mbZ \setminus \{0\}} \sum_{k=0}^K h_k \lambda_i^k \hat{X}(\lambda_i)\varphi_i(v)\ .
\end{equation*}
Hence, the frequency response of $T_\bbH$ is given by
\begin{equation} \label{eqn:transfer_fc_linear}
\hat{T}_\bbH(\lambda) = h(\lambda) =  \sum_{k=0}^K h_k \lambda^k\ .
\end{equation}

Equation \eqref{eqn:transfer_fc_linear} is the infinite counterpart of the frequency response of a LSI graph filter \eqref{eqn:any_graph_filter_resp}.
Note that, to understand the behavior of this filter on a specific graphon $\bbW$, we need to evaluate $h(\lambda)$ at each graphon eigenvalue $\lambda_i$.
But \eqref{eqn:transfer_fc_linear} is otherwise independent of the graphon. In other words, the frequency response of a graphon filter always has the same shape, irrespective of the graphon.
\blue{A third important remark pertaining to \eqref{eqn:transfer_fc_linear} is that LSI graphon filters can approximate any filter with \red{analytic} frequency response $h(\lambda)$ arbitrarily well as~$K\to\infty$}. This is because the frequency response of a LSI graphon filter \eqref{eqn:transfer_fc_linear} is a polynomial of the eigenvalues of the graphon. Put formally, a graphon filter with frequency response $\hat{T}_\bbH(\lambda) = h(\lambda)$ can be written as a LSI graphon filter (Def. \ref{defn:linear}) provided that $h(\lambda)$ \red{is analytic, i.e., that it is infinitely differentiable at $\{\lambda_i\}_{i \in \mbZ \setminus \{0\}}$ and its Taylor series converges pointwise}. 

We conclude this section by stressing that the goal of Defs. \ref{defn:linear} and \ref{defn:wft}, as well as of the definition of a graphon signal, is to \textit{generalize} GSP concepts to graphons. 
These concepts are not realizable in the way that graph signals, graph filters, and the GFT are because, unlike graphs, graphons are intangible theoretical objects. Nonetheless, their value lies in that they help understand the behavior of graph signals in the limit of large-scale networks. This provides the theoretical foundations to enable the practical scenarios (\textbf{S1}--\textbf{S3}) discussed in the introduction (see also Sec. \ref{sec:sims}). Indeed, as we show next, the WFT and the LSI graphon filter are mathematical limits of the GFT and of the LSI graph filter on convergent sequences of graph signals.

\begin{remark}[Uniqueness of limit graphon signal on $\bbW$] \label{remark:limit_uniqueness}
A sequence of graphs $\{\bbG_n\}$ converges to a graphon $\bbW$ if and only if the homomorphism densities $t(\bbF,\bbG_n)$ converge to $t(\bbF,\bbW)$ for every motif $\bbF$. Naturally, there may be other graphons $\bbW'$ with same homomorphism densities $t(\bbF,\bbW') = t(\bbF,\bbW)$ for all graphs $\bbF$ and so the limit graphon $\bbW$ is not necessarily unique, but in this paper we select one of these limits---the graphon $\bbW$--- without loss of generality\footnote{Graphons $\bbW'$ and $\bbW$ with same homomorphism densities $t(\bbF,\bbW')=t(\bbF,\bbW)$ for all $\bbF$ are called weakly isomorphic. We can select $\bbW$ without loss of generality because two graphons $\bbW'$ and $\bbW$ are weakly isomorphic if and only if $\delta_{\square}(\bbW',\bbW) = 0$, i.e., if their cut distance is zero  \cite[Chapters 7.3, 8.2.2]{lovasz2012large}. Therefore, the limit graphon is unique w.r.t. the metric induced by the cut distance in the space of \textit{unlabeled} graphons.} and use its implicit node labeling to define sequences of graph signals $\{(\bbG_n, \bbx_n)\}$ that converge to graphon signals $(\bbW,X)$ [cf. Def. \ref{defn:graph_signal_conv}]. Since $\bbW$ is fixed, for each sequence $\{(\bbG_n, \bbx_n)\}$ the limit signal $(\bbW,X)$ is unique in $L^2$.
\end{remark}


%% file: plots_stability/eigenvalues_graphon2.tex

\usetikzlibrary{arrows, arrows.meta}
\usetikzlibrary{decorations.markings}

\pgfplotsset{xtick style={draw=none}}

\def \thisplotscale {3.4}
\def \unit {\thisplotscale cm}

\def \frequencyresponse{0.75}

\begin{tikzpicture}[x = 1*\unit, y=1*\unit]

\path[fill=black, opacity = 0.1] 
              (0.9, 0.0) --
              (0.9, 1.0) --  
              (1.5, 1.0) -- 
              (1.5, 0.0) -- cycle;

\begin{axis}[scale only axis,
             width  = 2.4*\unit,
             height = 1*\unit,
             xmin = -8, xmax=8,
             tick style={opacity=1},
             xtick = {-8, -7, -5, -3, -2, 0, 2, 2.7, 4.24, 5.51, 7.6, 8},
             xticklabels = {\black{\tiny $-1\ \ \ $},
             			   \blue{\tiny $\ \lam_{-1}$}, 
                            \blue{\tiny $\ \lam_{-2}$},
                            \blue{\tiny $\cdot \cdot \cdot$}, 
                            \black{\tiny $-c$},
             			   \black{\tiny $0$}, 
			                \black{\tiny $c$},
                            \blue{\tiny $\ \cdot \cdot \cdot$}, 
                            \blue{\tiny $\ \lam_3$},
                            \blue{\tiny $\ \lam_2$}, 
                            \blue{\tiny $\ \lam_1$},
                            \black{\tiny $\ \ \ 1$}},
             ymin = -0, ymax = 1.15,
             ytick = {-1},
             enlarge x limits=false]

\addplot+[samples at = {-7, -5, -3, -1.4, -0.7, -0.44, -0.33, -0.22, -0.1, -0.04, -0.02, 0.00, 0.02, 0.04, 0.1, 0.2, 0.3, 0.45, 0.72, 0.98, 1.2, 2.7, 4.24, 5.51, 7.6}, 
          color = blue!60, 
          ycomb, 
          mark=otimes*, 
          mark options={blue!60}]
         {\frequencyresponse};

\addplot+[samples at = {-8,-2,0,2, 8}, 
          color = black!70, 
          line width = 1.05,
          ycomb, no marks]
         {0.05};


\end{axis}

\draw[|-|,color = black, 
          line width = 0.65] (0.9,0.75) -- (0.99,0.75);
\draw[|-|,color = black, 
          line width = 0.65] (1.38,0.75) -- (1.5,0.75);
\node at (1.21, 0.93) {\scriptsize $d_n$};
\draw[decoration={markings,mark=at position 1 with {\arrow[scale=0.5]{>}}},
    postaction={decorate},
    shorten >=0.1pt] (0.945, 0.8) -- (1.13,0.88);
\draw[decoration={markings,mark=at position 1 with {\arrow[scale=0.5]{>}}},
    postaction={decorate},
    shorten >=0.1pt] (1.44, 0.8) -- (1.27,0.88);

\end{tikzpicture}


%% file: plots_stability/eigenvalues_graph_graphon.tex

\pgfplotsset{xtick style={draw=none}}

\def \thisplotscale {3.4}
\def \unit {\thisplotscale cm}

\def \frequencyresponse 
     {0.8}

\begin{tikzpicture}[x = 1*\unit, y=1*\unit]
\begin{axis}[scale only axis,
             width  = 2.4*\unit,
             height = 1*\unit,
             xmin = 0, xmax=8,
             xtick = {0, 4, 8},
             xticklabels = {\black{\footnotesize $0$},
             				\black{\footnotesize $\lambda$}, 
                            \black{\footnotesize $1$}},
             ymin = -0, ymax = 1.15,
             ytick = {-1},
             typeset ticklabels with strut,
             enlarge x limits=false,
             legend pos=north east,
             legend columns=3]


\addplot+[samples at = {0.00, 0.02, 0.04, 0.1, 
			  0.2, 0.3, 0.45, 0.6, 0.79, 1.03, 1.57, 
              2.63, 3, 4.24, 5.51, 6.72, 7.6}, 
          color = blue!60, 
          ycomb, 
          mark=otimes*, 
          mark options={blue!60}]
         {\frequencyresponse};
         
\addplot+[samples at = {0.04, 0.175, 0.155, 0.34, 0.355, 
              0.66, 0.72, 1.17, 1.4, 2.5, 3.13, 4.13, 
              5.68, 6.87, 7.8}, 
          color = red!60, 
          ycomb, 
          mark=otimes*, 
          mark options={red!60}]
         {\frequencyresponse};

   \addlegendentry{\scriptsize $\bbW$}
   \addlegendentry{\scriptsize $\bbG_n$}


\end{axis}
\end{tikzpicture}


%% file: filter-graphon.tex

%
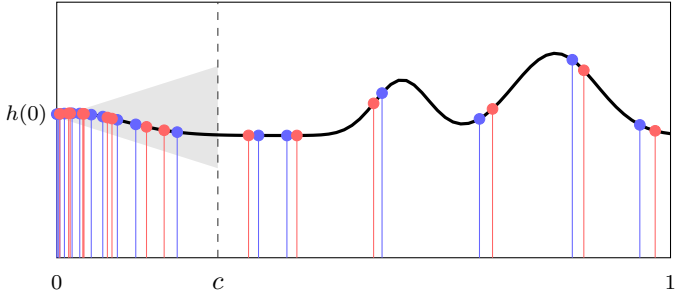
\begin{figure}[t]
    \centering
    \input{plots_stability/filter_response.tex}    
    \caption{Lipschitz continuous filter function $h(\lambda)$ with Lipschitz constant $L$. Only the positive eigenvalue axis is depicted for simplicity. Lipschitz filters eliminate the requirement that the graphon signal be bandlimited because they bound the filter variation for signal components associated with eigenvalues smaller than $c$.}
    \label{fig:filter_response}
\end{figure}


In this section, we leverage the properties of convergent graph sequences to prove a series of convergence results which show that GSP converges to WSP. Our first result describes the limit behavior of the GFT on convergent sequences of graph signals (Sec. \ref{sbs:gft}). We show that, when the limit graphon signal is bandlimited, the GFT converges to the WFT (Thm. \ref{thm:wft}). 
Our second result shows that the spectral responses of graph filters converges to the spectral response of a well-defined graphon filter (Thm. \ref{thm:transfer_fcn}, Sec. \ref{sbs:filter}). Thm. 2  is further extended to show that convergence of the graph signal and graph filter also implies convergence of the filter response in the \textit{vertex domain}. This is our third convergence result. It is stated with increasing levels of generality in Cor. \ref{shift_filter_response}, which follows directly from convergence of the GFT and of the filter spectral responses; and in Thms. \ref{thm:non_disc} and \ref{thm:non_disc2}, which do not require the graphon signal to be bandlimited.
These findings, particularly the more general Thm. \ref{thm:non_disc2}, are the main technical contributions of this work. 
At the end of the section, a table summary of the GSP and WSP definitions of a signal, of the Fourier Transform and of linear shift-invariant filters can be found in Table \ref{tb:summary}.} This table also highlights the relationships between these concepts as established by Thms. \ref{thm:wft}--\ref{thm:non_disc2}.  

\subsection{Convergence of the GFT} \label{sbs:gft}

When a sequence of graph signals converges to a bandlimited graphon signal, we can show that the GFT converges to the WFT as long as the limit graphon is \textit{non-derogatory} (Def. \ref{defn:non_derog}). This is the main result of this section, presented in Thm. \ref{thm:wft}.

\begin{definition} \label{defn:non_derog}
A graphon $\bbW$ is non-derogatory if $\lambda_i \neq \lambda_j$ for all $i \neq j$ and $i, j \in \mbZ \setminus \{0\}$.
\end{definition}

\begin{theorem}[GFT $\to$ WFT for BL graphon signals] \label{thm:wft}
Let $\{(\bbG_n,\bbx_n)\}$ be a sequence of graph signals and let $\{\pi_n\} \in \ccalP$ be a sequence of permutations such that $\{(\bbG_n,\bbx_n)\}$ converges to the $c$-bandlimited graphon signal $(\bbW,X)$ in the sense of Def. \ref{defn:graph_signal_conv}, where $\bbW$ is non-derogatory. Then, 
\begin{align*}
\mbox{GFT}\left[(\pi_n(\bbG_n),\pi_n(\bbx_n))\right] \to \mbox{WFT}\left[(\bbW,X)\right]\ .
\end{align*}
Conversely, if $\{(\bbG_n,\hbx_n)\}$ is a sequence of GFTs converging to the WFT $(\bbW,\hat{X})$, then there exists a sequence of permutations $\{\pi_n\} \in \ccalP$ such that
\begin{align*}
\pi_n\bigg(\mbox{iGFT}\big[\bbx_n\big]\bigg) \to \mbox{iWFT}\big[\hat{X}\big].
\end{align*}
\end{theorem}

Note that the GFT coefficients $[\hbx]_i$ and the WFT coefficients $[\hat{X}]_i$ inherit the ordering of the eigenvalue sets of the graphon $\bbW$ and of the graphs $\bbG_n$, which are both ordered with indices $i \in \mbZ\setminus \{0\}$ according to their sign and in decreasing order of absolute value.

Thm. \ref{thm:wft} relates the GFT, a Fourier transform for ``discrete'' graph signals, to the WFT, a Fourier transform for ``continuous'' graphon signals. This comparison is only possible because, like graphs, graphons have discrete spectra.
This unveils an interesting parallel with the relationship between the discrete Fourier transform (DFT) --- a discrete transform for discrete signals --- and the Fourier series ---- a discrete transform for continuous signals --- in Euclidean domains. Another important consequence of Thm. \ref{thm:wft} is that it allows inferring  the spectral content of graph signals by analyzing the spectral content of the graphon signals giving rise to them (and vice-versa). This is useful, for instance, when these signals and/or the graphs on which they are supported are corrupted or only partially observable, in which case the WFT (or the GFT on another graph of the same family) can be used as an approximation of the GFT on the original graph.

We also point out that the requirement that the graphon be non-derogatory is not very restrictive: as stated in the following proposition, the space of non-derogatory graphons is dense in the space of graphons.

\begin{proposition}[Density of $\mathfrak{W}$] \label{prop:non_derog} 
Let $\mathfrak{W}$ denote the space of all bounded symmetric measurable functions $\bbW: [0,1]^2 \to \reals$, i.e., the space of graphons. The space of \emph{non-derogatory} graphons is dense in $\mathfrak{W}$.
\end{proposition}
\begin{proof}
Refer to Appendix D.
\end{proof}

Prop. \ref{prop:non_derog} tells us that, even if a graphon is derogatory, there exists a non-derogatory graphon arbitrarily close to it for which the GFT convergence result from Thm. \ref{thm:wft} holds.

In order to establish Thm. \ref{thm:wft} and subsequent results, we will work with the graphon signals induced by the graph signals we consider. We have already described in \eqref{eqn:graphon_signals_induced} how their (vertex domain) values are related. In the sequel, we will also need to relate their spectral properties. This relationship is formalized in Lemma \ref{T:induced_graphon}. 
Note that, although the adjacencies $\bbS_n$ of the graphs $\bbG_n$ have a finite number of eigenvalues $\lambda_i(\bbS_n)$, we still associate the eigenvalue sign with its index and order the eigenvalues in decreasing order of absolute value. The indices $i$ are now defined on some finite set $\ccalL \subseteq \mbZ \setminus \{0\}$.

\begin{lemma}\label{T:induced_graphon}
Let $(\bbW_\bbG,X_\bbG)$ be the graphon signal induced by the  graph signal $(\bbG,\bbx)$ on $n$ nodes. Then, for $i \in \ccalL$ we have
\begin{align*}
	\lambda_i(T_{\bbW_\bbG}) &= \frac{\lambda_{i}(\bbS)}{n}
	\\
	\varphi_i(T_{\bbW_\bbG})(u) &= [\bbv_i]_{j}
		\times \sqrt{n} \mbI\left( u \in I_j \right), \ j = 0, \ldots, n \mbox{ and}
	\\
	[\hat{X}_\bbG]_i &= \dfrac{[\hbx]_i}{\sqrt{n}}
\end{align*}
where $\lambda_i(\bbS)$ are the eigenvalues of the adjacency matrix. For $i \notin \ccalL$, we let $\lambda_i(T_{\bbW_\bbG}) = [\hat{X}_\bbG]_i = 0$ and $\varphi_i(T_{\bbW_\bbG}) = \psi_i$ such that $\{\varphi_i(T_{\bbW_\bbG})\} \cup \{\psi_i\}$ forms an orthonormal basis of $L^2([0,1])$.
\end{lemma}
\begin{proof} \renewcommand{\qedsymbol}{}
Refer to Appendix \ref{appx:lemma2}.
\end{proof}

\begin{proof}[Proof of Thm. \ref{thm:wft}]
We now prove that, since the finite set $\ccalL$ converges to $\mbZ \setminus \{0\}$ as $n$ goes to infinity, $\text{WFT}\{(\bbW_{\pi_n(\bbG_n)},\pi_n(X_{\bbG_n}))\}  \to \text{WFT}\{(\bbW,X)\}$. We leave the dependence of the convergent signal sequence $\{(\bbG_n,\bbx_n)\}$ on $\{\pi_n\}$ implicit and write $\bbW_n = \bbW_{\pi_n(\bbG_n)}$ and $X_n = X_{\pi_n(\bbG_n)}$. Next, we use the eigenvector convergence result from the following lemma. Thm. \ref{thm:wft} then follows from the fact that inner products are continuous in the product topology that they induce.
\begin{lemma}\label{T:eigenvectors_convergence}
Let $\ccalC = \{i \in \mbZ \setminus \{0\} \mid |\lambda_i(T_{\bbW})| \geq c\}$ be the set of indices of the non-vanishing eigenvalues and denote $\ccalS$ the subspace spanned by the eigenfunctions $\{\varphi_i(T_\bbW)\}_{i \notin \ccalC}$. Then,
\begin{enumerate}
\item[(i)] for $i \in \ccalC$, $\varphi_i(T_{\bbW_n}) \to \varphi_i(T_\bbW)$ in $L_2$, and
\item[(ii)] for $i \notin \ccalC$, $\varphi_i(T_{\bbW_n}) \to \red{\Psi_i} \in \ccalS$ \red{weakly}.
\end{enumerate}
\end{lemma}
\begin{proof} \renewcommand{\qedsymbol}{}
Refer to Appendix \ref{appx:lemma3}.
\end{proof}
Starting with the eigenvectors with indices in $\ccalC$, for any $\epsilon > 0$ it holds from Lemma \ref{T:eigenvectors_convergence} and from the convergence of $X_n$ in $L^2$ that there exist $n_1$ and $n_2$ such that
\begin{align*}
	\|\varphi_i(T_{\bbW_n}) - \varphi_i(T_\bbW)\|  \leq \frac{\epsilon}{2\|X\| }
		\text{, for all } n > n_1
	\\ \text{and} \quad
	\|X_n - X\|  \leq \frac{\epsilon}{2}
		\text{, for all } n > n_2\ .
\end{align*}
Recall that $\|\varphi_i(T_{\bbW_n})\|  \leq 1$ for all $n$ and $i \in \ccalC$ because the graphon spectral basis is orthonormal. Since the sequence $\{X_n\}$ is convergent, it is bounded and $\|X\|  < \infty$. Let $m = \max{\{n_1,n_2\}}$. Then, it holds that
\begin{align*}
	|[\hat{X}_{n}]_i &- [\hat{X}]_i|
    = |\langle{X_n},{\varphi_i(T_{\bbW_n})}\rangle - \langle{X},{\varphi_i(T_\bbW)}\rangle|
	\\
	{}&= |\langle{X_n - X},{\varphi_i(T_{\bbW_n})}\rangle + \langle{X},{\varphi_i(T_{\bbW_n}) - \varphi_i(T_\bbW)}\rangle|
	\\
	{}&\leq \|X_n - X\| \|\varphi_i(T_{\bbW_n})\| + \|X\|\|\varphi_i(T_{\bbW_n}) - \varphi_i(T_\bbW)\|
	\\
	{}&\leq \frac{\epsilon}{2} \|\varphi_i(T_{\bbW_n})\| + \|X\|\frac{\epsilon}{2\|X\|} \leq \epsilon \mbox{ for all }n > m\text{.}
\end{align*}

For $i \notin \ccalC$, the eigenfunctions $\varphi_i(T_{\bbW_n})$ may not converge to $\varphi_i(T_\bbW)$, but they do converge to some function $\red{\Psi_i} \in \ccalS$. Given that the graphon signal $(\bbW,X)$ is $c$-bandlimited, we have $\langle{X},{\varphi_i(T_\bbW)}\rangle = 0$ for $i \notin \ccalC$, so that $X$ must be orthogonal to all functions in $\ccalS$. 
Using the same argument as for $i \in \ccalC$ yields that the remaining GFT coefficients also converge to the WFT. Formally,
\begin{equation*}
	\langle{\varphi_i(T_{\bbW_n})},{X_n}\rangle \to \langle{\red{\Psi_i}},{X}\rangle = 0 = \langle{\varphi_i(T_\bbW)},{X}\rangle
		\text{ .}
\end{equation*}

Convergence of the iGFT to the iWFT follows directly from these results and from Lemma \ref{T:eigenvectors_convergence}. Explicitly, use the triangle inequality to write
\begin{align*} 
\begin{split}
\left\|\sum_{i \in \mbZ\setminus\{0\}} [\hat{X}]_i \varphi_i(T_\bbW) - \sum_{i \in \mbZ\setminus\{0\}} [\hat{X}_n]_i \varphi_i(T_{\bbW_n})\right\| \\
\leq \sum_{i \in \mbZ\setminus\{0\}} \|[\hat{X}]_i \varphi_i(T_\bbW) - [\hat{X}]_i \varphi_i(T_{\bbW_n})\| \\
+ \sum_{i \in \mbZ\setminus\{0\}}\|[\hat{X}]_i \varphi_i(T_{\bbW_n}) - [\hat{X}_n]_i \varphi_i(T_{\bbW_n})\| \ .
\end{split}
\end{align*}
Applying the Cauchy-Schwarz inequality and splitting the sums between $i \in \ccalC$ and $i \notin \ccalC$, we get
\begin{multline} \label{eqn:igft_conv_proof}
\left\|\sum_{i \in \mbZ\setminus\{0\}} [\hat{X}]_i \varphi_i(T_\bbW) - \sum_{i \in \mbZ\setminus\{0\}} [\hat{X}_n]_i \varphi_i(T_{\bbW_n})\right\| \\
\leq \sum_{i \in \ccalC} |[\hat{X}]_i |\|\varphi_i(T_\bbW) - \varphi_i(T_{\bbW_n})\| \\
+ \sum_{i \in \ccalC}|[\hat{X}]_i  - [\hat{X}_n]_i| \|\varphi_i(T_{\bbW_n})\| \\
+ \sum_{i \notin \ccalC}|[\hat{X}_n]_i| \|\varphi_i(T_{\bbW_n})\|\to 0\text{ .} \qedhere
\end{multline}
The first term on the right hand side of \eqref{eqn:igft_conv_proof} vanishes because $\varphi_i(T_{\bbW_n}) \to \varphi_i(T_\bbW)$ for $i \in \ccalC$ by Lemma \ref{T:eigenvectors_convergence}; the second term, because the GFT coefficients $[\hat{X}_n]_i$ converge to $[\hat{X}]_i$ for $i \in \ccalC$; and the third term, because $[\hat{X}_n]_i \to [\hat{X}]_i = 0$ for $i \notin \ccalC$.
\end{proof}

\subsection{Convergence of graph filter responses in the spectral and vertex domains} \label{sbs:filter}

Our second convergence result involves the frequency response of graph filters. As we have already noted, the frequency responses of LSI graph filters \eqref{eqn:any_graph_filter_resp} and of their graphon counterparts \eqref{eqn:transfer_fc_linear} have the same expression as long as the coefficients $h_k$ (or, equivalently, the function $h$) are the same. In the following, we show that these frequency responses actually converge to one another as the number of nodes goes to infinity. 

\begin{theorem}[Convergence of graph filter frequency response] \label{thm:transfer_fcn}
On the graph sequence $\{\bbG_n\}$, let $\bbH_n(\bbS_n)$ be a sequence of filters of the form $\bbH_n(\bbS_n) = \bbV_n h(\bbLam (\bbS_n) /n) \bbV_n^\Hr$; and, on the graphon $\bbW$, define the filter 
$$(T_\bbH X)(v) = \sum_{i\in \mbZ \setminus \{0\}}h(\lambda_i(T_\bbW)) \hat{X}(\lambda_i)\varphi_i(v)\ .$$ 
If $\{\bbG_n\}\to\bbW$ and $h: [0,1] \to \reals$ is continuous, then
\begin{equation*}
\lim_{n \to \infty} \hbH_n(\lambda_i(\bbS_n)/n) = \hat{T}_\bbH(\lambda_i(T_\bbW))
\end{equation*}
\blue{where $\hbH_n$ and $\hat{T}_\bbH$ are the frequency responses of $\bbH_n$ and $T_\bbH$ respectively.}
\end{theorem}
\begin{proof}
This is a direct consequence of a result due to \cite[Thm. 6.7]{borgs2012convergent} and restated here as Lemma \ref{eigenvalue_conv}. 
\begin{lemma}[Eigenvalue convergence] \label{eigenvalue_conv}
Let $\{\bbG_n\}$ be a sequence of graphs with eigenvalues $\{\lambda_i(\bbS_n)\}_{i \in \mbZ \setminus \{0\}}$, and $\bbW$ a graphon with eigenvalues $\{\lambda_i(T_\bbW)\}_{i \in \mbZ \setminus\{0\}}$.
If $\{\bbG_n\}\to\bbW$, 
\begin{equation} \label{eqn:eigenvalue_conv}
\lim_{n \to \infty} \dfrac{\lambda_i(\bbS_n)}{n} = \lim_{n \to \infty} \lambda_i(T_{\bbW_{\bbG_n}}) = \lambda_i(T_\bbW) \mbox{ for all } i\ .
\end{equation}
\end{lemma}
\begin{proof}
Refer to Appendix B.
\end{proof}
Lemma \ref{eigenvalue_conv} tells us that, in any convergent graph sequence, the eigenvalues of the graph converge to the eigenvalues of the limit graphon under a $1/n$ rescaling. This is illustrated in Fig. \ref{fig:eigenvalues_graph_graphon} for $\lambda > 0$.
Using Lemma \ref{eigenvalue_conv}, we will show that the transfer functions of arbitrary graph filters $\bbH_n(\bbS_n)$ converge to the transfer function of the graphon filter $T_\bbH$ with same filter function $h(\lambda)$. 

Consider the graphon signal $(\bbW,X)$. Applying $T_\bbH$ to $(\bbW, X)$ as in \eqref{eqn:lsi-wf}, we get
\begin{equation} \label{pf_resp_1}
Y(v) = \sum_{i \in \mbZ \setminus \{0\}} h(\lambda_i)\hat{X}(\lambda_i)\varphi_i(v)
\end{equation}
where we have omitted the dependence on $T_\bbW$ by writing $\lambda_i = \lambda_i(T_\bbW)$. The WFT of $(\bbW,Y)$ is given by
\begin{align*}
\begin{split}
[\hat{Y}]_j &= \int_0^1 Y(v) \varphi_j(v) dv\\
&= \sum_{i \in \mbZ \setminus \{0\}} h(\lambda_i)\bigg( \int_0^1 \varphi_i(v)\varphi_j(v) dv \bigg) \hat{X}(\lambda_i) \\
&= h(\lambda_j )\hat{X}(\lambda_j)
\end{split}
\end{align*}
from which we conclude that $\hat{T}_\bbH (\lambda_j) = h(\lambda_j)$.

We now determine the frequency response of $\bbH_n(\bbS_n)$. Applying $\bbH_n(\bbS_n)$ to the graph signal $(\bbG_n,\bbx_n)$, we get
\begin{align*}
\begin{split}
\bby_n = \bbH_n(\bbS_n) \bbx_n &= \bbV_n h(\bbLam (\bbS_n) / n)  \bbV_n^\Hr \bbx_n\ \\
&= \bbV_n h(\bbLam(\bbS_n)/ n) \hbx_n\ .
\end{split}
\end{align*}
The GFT of $(\bbG_n,\bby_n)$ is given by
\begin{equation*}
[\hby_n]_j = \bbv_{nj}^\Hr \bbV_n h(\bbLam(\bbS_n) / n) \hbx_n = h(\lambda_j(\bbS_n) / n) [\hbx_n]_j
\end{equation*}
and therefore $\hbH_n(\lambda_j(\bbS_n)) = h(\lambda_j(\bbS_n) / n)$.

Since $\bbG_n \to \bbW$, from Lemma \ref{eigenvalue_conv} it holds that $\lambda_j(\bbS_n)/n \to \lambda_j$. Because $h$ is continuous, this implies $h(\lambda_j(\bbS_n)/n) \to h(\lambda_j)$, which completes the proof.
\end{proof}

The spectral or \textit{frequency response} of a graph filter thus converges to that of the corresponding graphon filter provided that the frequency response of the filter, $h$, is continuous. In particular, this is the case for polynomials, so that the frequency response induced by LSI graph filters \eqref{eqn:any_graph_filter_resp} converges to that of their corresponding graphon filters \eqref{eqn:transfer_fc_linear}.
To understand the importance of this result, suppose that we design a filter with a certain spectral behavior on the graphon; Thm. \ref{thm:transfer_fcn} tells us that the same spectral behavior can be expected from the application of this filter (or, more precisely, of the graph filter with same coefficients) on graphs sampled from the graphon. By a simple triangle inequality argument, it then follows that filters can be transferred between graphs associated with the same graphon to obtain a similar spectral behavior. This is the first evidence of graph filter \textit{transferability}.

Nevertheless, Thm. \ref{thm:transfer_fcn} has a limitation. It only gives account of the limit behavior of the graph filter response in the frequency domain, while graph filters operate in the node domain. 
To analyze the vertex domain behavior of graph filters, we start with the simple case of bandlimited signals. Putting together Thms. \ref{thm:wft} and \ref{thm:transfer_fcn}, we can show that, when the limit of the graph signals is bandlimited, the graph filter outputs converge in the vertex domain.

\begin{corollary}[Convergence of graph filter response for bandlimited graphon signals]
\label{shift_filter_response}
Let $\{\bbG_n\}$ be a sequence of graphs converging to the graphon $\bbW$, where $\bbW$ is non-derogatory. Let $\bbH_n(\bbS_n) = \bbV_n h(\bbLam (\bbS_n) /n)\bbV_n^\Hr$ be a sequence of filters on the graphs $\{\bbG_n\}$, and let $(T_{\bbH}X)(v)=\sum_{i \in \mbZ \setminus \{0\}} h(\lambda_i)\hat{X}(\lambda_i)\varphi_i(v)$ be a filter on the graphon $\bbW$. Consider the sequence of graph signals $\{(\bbG_n,\bbx_n)\}$ and let $\{\pi_n\}$ be  a sequence of permutations such that $\{(\bbG_n,\bbx_n)\} \to (\bbW,X)$ in the sense of Def. \ref{defn:graph_signal_conv}.
Then, $\bby_n = \bbH(\pi_n(\bbS_n))\pi_n(\bbx_n)$ converges to $Y = T_\bbH X$.
\end{corollary}
\begin{proof}
Without loss of generality, assume that the permutation sequence $\{\pi_n\}$ is a sequence of identity permutations, i.e., the labeling of the signals $\bbx_n$ already matches the labeling for which the graphs converge in cut norm. Let the WFT of $(\bbW,Y)$ be $[\hat{Y}]_i = \hat{T}_\bbH(\lambda_i) [\hat{X}]_i$ and the GFT of $(\bbG_n, \bby_n)$ be $[\hby_n]_i = \hbH_n(\lambda_i(\bbS_n)/n) [\hbx_n]_i$. By Thm. \ref{thm:wft}, $[\hbx_n]_i \to [\hat{X}]_i$, and, by Thm. \ref{thm:transfer_fcn}, $\hbH_n(\lambda_i(\bbS_n)/n) \to \hat{T}_\bbH(\lambda_i)$. Since $\hat{T}_\bbH$ is a linear operator, and hence continuous, $[\hby_n]_i \to [\hat{Y}]_i$. 
Applying Thm. \ref{thm:wft} once again for the iGFT, we conclude that $\bby_n \to Y$.
\end{proof}

Cor. \ref{shift_filter_response} extends upon Thm. \ref{thm:transfer_fcn} by showing that, provided that the sequence of input signals $\{(\bbG_n, \bbx_n)\}$ converges to a bandlimited graphon signal, the output signals obtained by applying the filters $\bbH_n(\bbS)$ converge in the same sense as $\{(\bbG_n, \bbx_n)\}$ in the vertex domain.  
The requirement that the graphon signal be bandlimited arises from the difficulty of matching the GFT and WFT coefficients associated with small eigenvalues, i.e., eigenvalues $\lambda_i$ for which $|i|$ is large. This is illustrated in Fig. \ref{fig:eigenvalues_graph_graphon}. Note that, as the eigenvalues approach 0, it becomes hard to tell which graph eigenvalue converges to which graphon eigenvalue, as the eigenvalue difference $\lambda_i - \lambda_{i+1}$ tends to zero as $i \to \infty$. 

This requirement can be eliminated by considering \textit{Lipschitz graph and graphon filters}, i.e., filters with Lipschitz continuous $h(\lambda)$. A function $h: [0,1] \to \reals$ is $L$-Lipschitz continuous if, for all $\lambda, \lambda' \in [0,1]$,
\begin{equation} \label{eqn:lipschitz_filters}
|h(\lambda)-h(\lambda')| \leq L |\lambda - \lambda'|\ .
\end{equation}
For $h$ differentiable, this is equivalent to bounding $dh/d\lambda$ by $L$ in absolute value. An example of Lipschitz continuous filter is shown in Fig. \ref{fig:filter_response}.
For filter functions $h$ satisfying \eqref{eqn:lipschitz_filters}, we can show that the graph filters converge in the vertex domain for \textit{any} graphon signal, not only bandlimited ones, because the variation of Lipschitz filters can be bounded close to zero [cf. Fig. \ref{fig:filter_response}].


\begin{theorem}[Convergence of filter response for Lipschitz continuous graph filters] \label{thm:non_disc}
Let $\{\bbG_n\}$ be a sequence of graphs converging to the graphon $\bbW$, where $\bbW$ is non-derogatory. Let $\bbH_n(\bbS_n) = \bbV_n h(\bbLam (\bbS_n) /n)\bbV_n^\Hr$ be a sequence of filters on the graphs $\{\bbG_n\}$, and let $(T_{\bbH}X)(v)=\sum_{i \in \mbZ \setminus \{0\}} h(\lambda_i)\hat{X}(\lambda_i)\varphi_i(v)$ be a filter on the graphon $\bbW$. 
Consider a sequence of graph signals $\{(\bbG_n,\bbx_n)\}$ 
and let $\{\pi_n\} \in \ccalP$ be a sequence of permutations such that $\{(\bbG_n,\bbx_n)\} \to (\bbW,X)$ in the sense of Def. \ref{defn:graph_signal_conv}.
Then, $\bby_n = \bbH(\pi_n(\bbS_n))\pi_n(\bbx_n)$ converges to $Y = T_\bbH X$.
\end{theorem}


\begin{table}[]
\centering
\caption{\blue{Table summary of GSP and WSP.}}
\begin{tabular}{l|ccc}
\hline
& \blue{Graph}  & \blue{Graphon}  & \blue{Convergence result}   \\ \hline
\blue{Signal}   & \blue{$(\bbG,\bbx)$}  & \blue{$(\bbW,X)$}  & \blue{Def. \ref{defn:graph_signal_conv}} \\ 
\blue{FT} & \blue{$\hbx$ (eq. \eqref{eqn:gft})}        & \blue{$\hat{X}$ (Def. \ref{defn:wft})} & \blue{Thm. \ref{thm:wft}}   \\
\blue{Filter}            & \blue{$\bbH(\bbS)$ (eq. \eqref{eqn:lsi-gf})} & \blue{$T_\bbH$ (Def. \ref{defn:linear})}     & \blue{Thm. \ref{thm:non_disc2}} \\ \hline
\end{tabular}
\label{tb:summary}
\end{table}
\begin{proof}
To prove convergence of the $(\bbG_n,\bby_n)$ to $(\bbW,Y)$, we consider the graphon signals $(\bbW_{\pi_n(\bbG_n)}, X_{\pi_n(\bbG_n)})$ induced by the graph signals $(\pi_n(\bbG_n),\pi_n(\bbx_n))$. The spectral properties of these signals and of the corresponding graph signals are related through Lemma \ref{T:induced_graphon}. To simplify notation, we once again leave the dependence on $\pi_n(\bbG_n)$ implicit and write $\bbW_n = \bbW_{\pi_n(\bbG_n)}$ and $X_n = X_{\pi_n(\bbG_n)}$. We also denote the induced graphon eigenvalues $\lambda^n_i = \lambda_i(T_{\bbW_n})$. Recall that these are given by $\lambda^n_i=\lambda_i(\bbS_n)/n$ per Lemma \ref{T:induced_graphon}. 

Without loss of generality, consider the normalized filter function $\bar{h}(\lambda) = h(\lambda)/ \sup_{\lambda \in [0,1]} |h(\lambda)|$. The signal $(\bbW, Y)$ obtained by applying $T_{\bar{\bbH}}$ to $(\bbW,X)$ can be written as
\begin{equation} \label{eqn:spectral_proof1}
Y(v) = \sum_{i \in \mbZ \setminus\{0\}} \bar{h}(\lambda_i) \hat{X}(\lambda_i) \varphi_i(v)
\end{equation}
and $(\bbW_n,Y_n)$, which is induced by $\bby_n = \bar{\bbH}(\bbS_n)\bbx_n$, as
\begin{equation} \label{eqn:spectral_proof2}
Y_n(v) = \sum_{i \in \mbZ \setminus\{0\}} \bar{h}(\lambda^n_i) \hat{X}_n(\lambda^n_i) \varphi_i(T_{\bbW_n})(v)\ .
\end{equation}
The dependence of the eigenfunctions $\varphi_i(T_{\bbW_n})$ on $T_{\bbW_n}$ is made explicit to distinguish them from $\varphi_i$, the eigenvalues of $T_\bbW$.

To show that the $(\bbW_n, Y_n)$ converge to $(\bbW, Y)$, we start by writing their norm difference using \eqref{eqn:spectral_proof1} and \eqref{eqn:spectral_proof2},
\begin{align} \label{eqn:spectral_proof3}
\begin{split}
\|Y - &Y_n\| = \\ 
&\left\|\sum_{i \in \mbZ \setminus\{0\}} \bar{h}(\lambda_i) \hat{X}(\lambda_i) \varphi_i - 
\sum_{i \in \mbZ \setminus\{0\}} \bar{h}(\lambda_i^n) \hat{X}_n(\lambda_i^n) \varphi_i(T_{\bbW_n})\right\| \text{.}
\end{split} 
\end{align}

Defining the set $\ccalC = \{i\ |\ |\lambda_i| \geq c\}$ for $c={(1-|\bar{h}_0|)}/{L(2\|X\|\epsilon^{-1}+1)}$ with  $\epsilon > 0$ and $\bar{h}_{0}= \bar{h}(0)$, these sums can be split up between $i \in \ccalC$ and $i \notin \ccalC$ to yield
\begin{equation} \label{eqn:spectral_proof4}
\begin{split}
    &\left\| \sum_{i \in \mbZ \setminus\{0\}} \bar{h}(\lambda_i) \hat{X}(\lambda_i) \varphi_i - \sum_{i \in \mbZ \setminus\{0\}} \bar{h}(\lambda_i^n) \hat{X}_n(\lambda_i^n) \varphi_i(T_{\bbW_n})\right\| \\
    &\leq \left\|\sum_{i \in \ccalC} \bar{h}(\lambda_i) \hat{X}(\lambda_i) \varphi_i - \sum_{i \in \ccalC} \bar{h}(\lambda_i^n) \hat{X}_n(\lambda_i^n) \varphi_i(T_{\bbW_n}) \right\|\ \mbox{\textbf{(i)} } \\ 
    &+ \left\|\sum_{i \notin \ccalC} \bar{h}(\lambda_i) \hat{X}(\lambda_i) \varphi_i - \sum_{i \notin \ccalC} \bar{h}(\lambda_i^n) \hat{X}_n(\lambda_i^n) \varphi_i(T_{\bbW_n}) \right\|\ \mbox{\textbf{(ii)}}\ .
    \end{split}
\end{equation}

Note that \textbf{(i)} corresponds to the difference between two bandlimited  graphon signals. By Cor. \ref{shift_filter_response}, there exists $n_0$ such that, for all $n > n_0$,
\begin{equation} \label{eqn:spectral_proof5}
\left\|\sum_{i \in \ccalC} \bar{h}(\lambda_i) \hat{X}(\lambda_i) \varphi_i - \sum_{i \in \ccalC} \bar{h}(\lambda_i^n) \hat{X}_n(\lambda_i^n) \varphi_i(T_{\bbW_n}) \right\| < \epsilon\ .
\end{equation}

For \textbf{(ii)}, we use the filter's Lipschitz property and Cauchy-Schwarz to write
\begin{equation} \label{eqn:spectral_proof6}
\begin{split}
&\left\|\sum_{i \notin \ccalC} \bar{h}(\lambda_i) \hat{X}(\lambda_i) \varphi_i - \sum_{i \notin \ccalC} \bar{h}(\lambda_i^n) \hat{X}_n(\lambda_i^n) \varphi_i(T_{\bbW_n}) \right\|\  \\
&\leq \left\|\sum_{i \notin \ccalC} (\bar{h}_{0}+Lc) \hat{X}(\lambda_i) \varphi_i - \sum_{i \notin \ccalC} (\bar{h}_{0}-Lc) \hat{X}_n(\lambda_i^n) \varphi_i(T_{\bbW_n}) \right\|\
\\
&\leq |\bar{h}_{0}|\left\|\sum_{i \notin \ccalC} \left[\hat{X}(\lambda_i) \varphi_i - \hat{X}_n(\lambda_i^n)  \varphi_i(T_{\bbW_n})\right]\right\| \\
 &+ Lc\left\|\sum_{i \notin \ccalC} \hat{X}(\lambda_i) \varphi_i \right\| + Lc  \left\|\sum_{i \notin \ccalC} \hat{X}_n(\lambda_i^n) \varphi_i(T_{\bbW_n}) \right\| 
\end{split}
\end{equation}
where the last inequality follows from the triangle inequality.

Because $\{\varphi_i\}$ and $\{\varphi_i(T_{\bbW_n})\}$ form complete bases of $L^2$, $\sum_{i \notin \ccalC} \hat{X}(\lambda_i)\varphi_i$ and $\sum_{i \notin \ccalC} \hat{X}_n(\lambda_i^n)\varphi_i(T_{\bbW_n})$ can be written as
\begin{align} \label{eqn:identities1}
&\sum_{i \notin \ccalC} \hat{X}(\lambda_i)\varphi_i = X - \sum_{i \in \ccalC} \hat{X}(\lambda_i)\varphi_i \quad \mbox{and} \\ \label{eqn:identities2}
&\sum_{i \notin \ccalC} \hat{X}_n(\lambda_i^n)\varphi_i(T_{\bbW_n}) = X_n - \sum_{i \in \ccalC} \hat{X}(\lambda_i^n)\varphi_i(T_{\bbW_n})
\end{align}
i.e., as the difference between the input signal and a bandlimited signal.
Using these identities and the triangle inequality, we leverage the fact that $X_n\to X$ in $L^2$ and apply Thm. \ref{thm:wft} to show that there exists $n_1$ such that, for all $n > n_1$,

\begin{align} \label{eqn:spectral_proof7}
\begin{split}
&\left\|\sum_{i \notin \ccalC} \hat{X}(\lambda_i) \varphi_i - \hat{X}_n(\lambda_i^n)  \varphi_i(T_{\bbW_n})\right\| \\
&\leq \left\|X-X_n\right\| + \left\|\sum_{i \in \ccalC} \hat{X}_n(\lambda_i^n)  \varphi_i(T_{\bbW_n}) -  \hat{X}(\lambda_i) \varphi_i\right\| < \epsilon\ .
\end{split}
\end{align}

As for $\|\sum_{i \notin \ccalC} \hat{X}_n(\lambda_i^n) \varphi_i(T_{\bbW_n})\|$, we can use the identities in \eqref{eqn:identities1} and \eqref{eqn:identities2} together with the triangle inequality to write
\begin{align*}
\begin{split}
\left\|\sum_{i \notin \ccalC} \hat{X}_n(\lambda_i^n) \varphi_i(T_{\bbW_n})\right\| \leq \left\|X_n - X\right\| + \left\|\sum_{i \notin \ccalC} \hat{X}(\lambda_i)\varphi_i\right\| \\
+ \left\|\sum_{i \in \ccalC} \hat{X}(\lambda_i)\varphi_i- \sum_{i \in \ccalC} \hat{X}_n(\lambda_i^n)\varphi_i(T_{\bbW_n})\right\|\ .
\end{split}
\end{align*}
From Thm. \ref{thm:wft} and the fact that $X_n\to X$ in $L^2$,
\begin{align} \label{eqn:spectral_proof7.5}
\begin{split}
\left\|\sum_{i \notin \ccalC} \hat{X}_n(\lambda_i^n) \varphi_i(T_{\bbW_n})\right\|
\leq \epsilon + \left\|\sum_{i \notin \ccalC} \hat{X}(\lambda_i)\varphi_i\right\| \mbox{ for } n>n_1.
\end{split}
\end{align}

Applying the Cauchy-Schwarz and triangle inequalities and substituting \eqref{eqn:spectral_proof7} and \eqref{eqn:spectral_proof7.5} in \eqref{eqn:spectral_proof6}, we arrive at a bound for \textbf{(ii)},
\begin{equation} \label{eqn:spectral_proof8}
\begin{split}
\left\|\sum_{i \notin \ccalC} \bar{h}(\lambda_i) \hat{X}(\lambda_i) \varphi_i - \sum_{i \notin \ccalC} \bar{h}(\lambda_i^n) \hat{X}_n(\lambda_i^n) \varphi_i(T_{\bbW_n})\right\| \\
\leq (|\bar{h}_{0}|+Lc)\epsilon + 2Lc \left\|\sum_{i \notin \ccalC} \hat{X}(\lambda_i) \varphi_i\right\| \\
\leq (|\bar{h}_{0}|+Lc) \epsilon + 2Lc\|X\| = \epsilon\ .
\end{split}
\end{equation}
Putting \eqref{eqn:spectral_proof5} and \eqref{eqn:spectral_proof8} together, we have thus proved that for all $n > \max{\{n_0,n_1\}}$, $\|Y-Y_n\| < 2\epsilon$, i.e., the output of $\bar{\bbH}(\bbS_n)$ converges to the output of $T_{\bar{\bbH}}$ in the vertex domain.
\end{proof}

Thm. \ref{thm:non_disc} broadens the scope of Cor. \ref{shift_filter_response} by extending the filter response convergence result to sequences of graph signals converging to \textit{generic} finite energy graphon signals that are not necessarily bandlimited. The Lipschitz condition on the filter $h$ allows bounding the variability of the filter response for signal components associated with eigenvalues smaller than some threshold $c \in [0,1]$, which can be made arbitrarily small [cf. Fig. \ref{fig:filter_response}].

Thm. \ref{thm:non_disc} can be further generalized to \textit{any graphon} as opposed to only non-derogatory ones. The difference in the case of derogatory graphons is that the WFT cannot be defined, so Thm. \ref{thm:wft} cannot be used in the proof of Thm. \ref{thm:non_disc2}. The proof argument needed in this case is therefore slightly different. However, this is extenuated by Prop. \ref{subspace_conv}. As long as eigengaps between adjacent graphon eigenspaces can be defined, this proposition ensures convergence not only of the eigenvectors, but also of the finite-dimensional eigenspaces associated with the repeated eigenvalues of an arbitrary graphon.

\begin{proposition}[Graphon subspace convergence] \label{subspace_conv}
Let $\{\bbG_n\}$ be a sequence of graphs with eigenvalues $\lambda_i(\bbS_n)$ converging to the graphon $\bbW$ with eigenvalues $\lambda_i$. 
If a given $\lambda_i$ has multiplicity $m_i$ and $\lambda_{i_k}^n = \lambda_{i_k}(\bbS_n)/n$ are the eigenvalues of $\bbW_{\bbG_n}$ (i.e., of the graphon induced by $\bbG_n$) converging to $\lambda_i$ [cf. Lemma \ref{T:induced_graphon}], then there exists a sequence of permutations $\{\pi_n\} \in \ccalP$ such that
\begin{equation*}
\vertiii{E_{T_{\bbW_{\pi_n\left(\bbG_n\right)}}}({\{\lambda_{i_k}^n\}}) - E_{T_\bbW}({\lambda_i})} \to 0
\end{equation*}
where $\ccalP$ is the set of admissible permutation sequences for the sequence $\{\bbG_n\}$ (Def. \ref{def:permutation_sequence}) and $E_{T}(\Lambda)$ is the projection operator onto the subspace associated with the eigenvalues in the set $\Lambda$ of the operator $T$.
\end{proposition}
\begin{proof}
Refer to Appendix \ref{appx:prop2}.
\end{proof}

With Prop. \ref{subspace_conv}, we are now equipped to state our most general result: vertex domain convergence of Lipschitz continuous graph filters for graph sequences converging to arbitrary graphons. This result is presented in Thm. \ref{thm:non_disc2}. We defer the proof to the appendices.

\begin{theorem}[Convergence of filter response for Lipschitz continuous graph filters] \label{thm:non_disc2}
Let $\{\bbG_n\}$ be a sequence of graphs converging to the graphon $\bbW$. Let $\bbH_n(\bbS_n) = \bbV_n h(\bbLam (\bbS_n) /n)\bbV_n^\Hr$ be a sequence of filters on the graphs $\{\bbG_n\}$, and let $(T_{\bbH}X)(v)=\sum_{i \in \mbZ \setminus \{0\}} h(\lambda_i)\hat{X}(\lambda_i)\varphi_i(v)$ be a filter on the graphon $\bbW$. Consider a sequence of graph signals $\{(\bbG_n,\bbx_n)\}$ 
and let $\{\pi_n\} \in \ccalP$ be a sequence of permutations such that $\{(\bbG_n,\bbx_n)\} \to (\bbW,X)$ in the sense of Def. \ref{defn:graph_signal_conv}.
Then, $\bby_n = \bbH(\pi_n(\bbS_n))\pi_n(\bbx_n)$ converges to $Y = T_\bbH X$.
\end{theorem}
\begin{proof} Refer to Appendix \ref{appx:thm4}.
\end{proof}

\blue{
The main takeaway from Thms. \ref{thm:non_disc} and \ref{thm:non_disc2} is that, if the limit graphon is known, we can trade the design of multiple filters in different graphs by the centralized design of a single graphon filter from which graph filters can then be sampled.
In practice, a more relevant implication of these theorems is that graph filters can be \textit{transferred} across graphs associated with the same graphon.
The ability to transfer graph filters is especially important when graphs are large or dynamic, as the operations involved in designing filters for these graphs can come out costly. This property is also inherited by graph neural networks (GNNs) based on these graph filters \cite{ruiz20-transf}. Transferability of GNNs has been demonstrated empirically in a number of applications \cite{eisen2019optimal,tolstaya2020learning}, and is formally characterized in \cite{ruiz2020gnns}, where transferability bounds are derived for both GNNs and graph filters.
Transferability of graph filters will also be illustrated in the numerical experiments of Sec. \ref{sec:sims}.

\blue{
\begin{remark}
Note that, while the results presented in Thms. \ref{thm:wft}--\ref{thm:non_disc2} may appear intuitive, their proofs are not. For instance, our Fourier convergence theorem (Thm. \ref{thm:wft}) requires that the graph and graphon signals be bandlimited for the GFT to converge to the WFT. This is in constrast to classical signal processing, where for any convergent sequence of length-$n$ discrete time signals on $[0,1]$ the discrete Fourier transform (DFT) converges to the Fourier transform (FT) regardless of the underlying spectral properties. This occurs because the regular line graphs underlying these signals have spectra that are evenly distributed on $[-1,1]$ and therefore never accumulate around zero. 
Unexpectedly, however, these conditions are not needed to show convergence of graph filter outputs. 
Indeed, while one would expect that graph filter outputs converge only for bandlimited signals, this is not the case in Thms. \ref{thm:non_disc}--\ref{thm:non_disc2}. Instead, these theorems require the filter to be Lipschitz for $|\lambda| < c$ [cf. Fig. \ref{fig:filter_response}]. 
This arises from the fact that, for small $\lambda$, the graph eigenspaces can become
hard to match to the corresponding graphon eigenspaces since the eigenvalues of the latter accumulate near zero. We can therefore replace bandlimitedness by a filter regularity condition.
\end{remark}
}

%% file: plots_stability/filter_response.tex

\pgfplotsset{xtick style={draw=none}}
\pgfplotsset{ytick style={draw=none}}

\def \thisplotscale {3.4}
\def \unit {\thisplotscale cm}

\def \frequencyresponse 
     { 0.1*exp(-(1*(x-0.2))^2) 
       + 0.25*exp(-(2*(x-4.5))^2) 
       + 0.37*exp(-(1.3*(x-6.5))^2) 
       + 0.55}

\begin{tikzpicture}[x = 1*\unit, y=1*\unit]
\def \factorx {2.4/8}
\def \deltax  {0.5*\factorx}
\def \shadeshift  {0.05}

\path [fill=black, opacity = 0.1] 
              (0.0, 0.55) --
              (0.63, 0.75) --  
              (0.63, 0.35) -- cycle;

\begin{axis}[scale only axis,
             width  = 2.4*\unit,
             height = 1*\unit,
             xmin = 0, xmax=8,
             xtick = {0, 2.1, 8},
             xticklabels = {\black{\footnotesize $0$},
             				\black\footnotesize {$c$}, 
                            \black{\footnotesize $1$}},
             ymin = -0, ymax = 1.15,
             ytick = {0.65},
             yticklabels = {\black{\footnotesize $h(0)$}},
             typeset ticklabels with strut,
             enlarge x limits=false]

\addplot+[samples at = {0.00, 0.02, 0.04, 0.1, 		              
			  0.2, 0.3, 0.45, 0.6, 0.79, 1.03, 1.57, 
              2.63, 3, 4.24, 5.51, 6.72, 7.6}, 
          color = blue!60, 
          ycomb, 
          mark=otimes*, 
          mark options={blue!60}]
         {\frequencyresponse};
         
\addplot+[samples at = {0.04, 0.175, 0.155, 0.34, 0.355, 
              0.66, 0.72, 1.17, 1.4, 2.5, 3.13, 4.13, 
              5.68, 6.87, 7.8}, 
          color = red!60, 
          ycomb, 
          mark=otimes*, 
          mark options={red!60}]
         {\frequencyresponse};
         
\addplot+[samples at = {2.1}, 
          color = black!80, 
          dashed,
          ycomb,
          mark=none]
         {1.2};

\addplot[ domain=0:8, 
          samples = 80, 
          color = black,
          line width = 1.2]
         {\frequencyresponse};

         

\end{axis}
\end{tikzpicture}


%% file: sims-graphon.tex


\begin{figure}
  \centering
  \includegraphics[width=0.92\columnwidth]{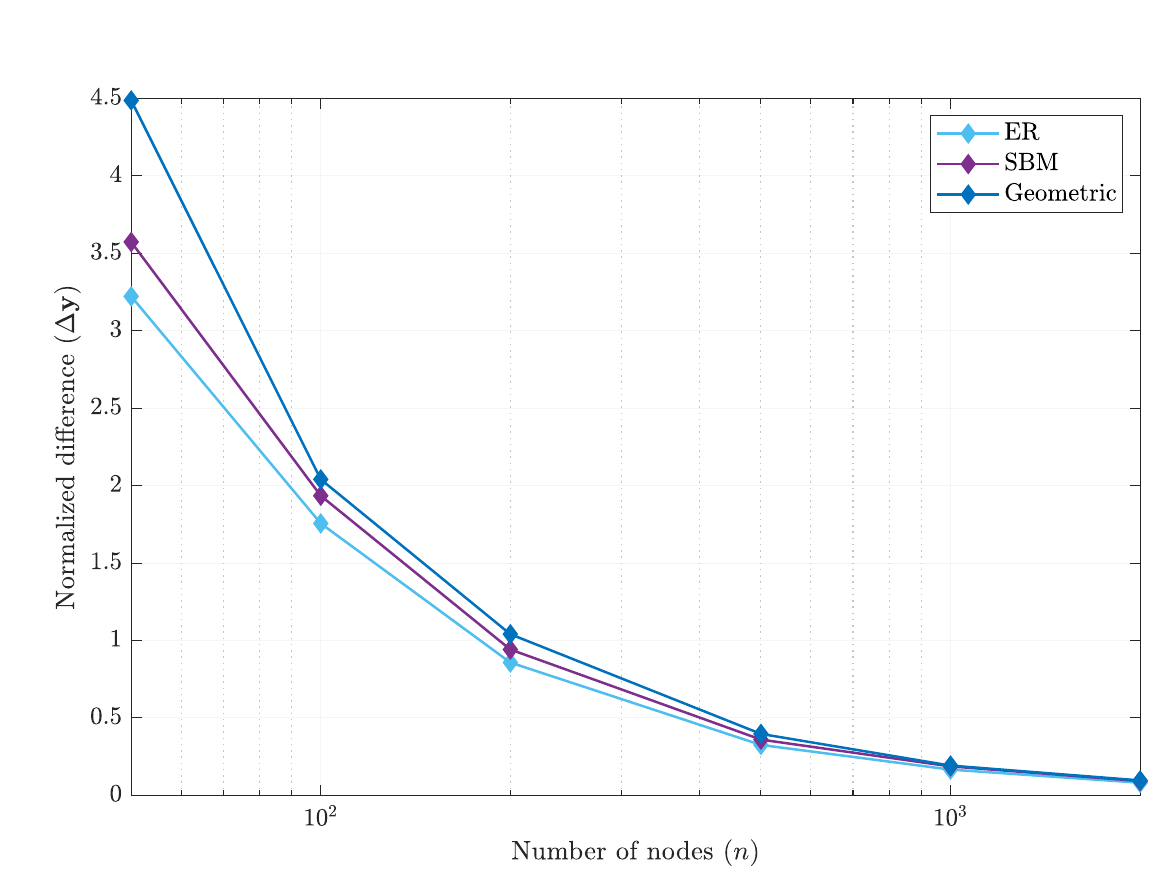}
  \caption{Norm difference between GMRF graphon signals diffused on ER, SBM and geometric graphons and the corresponding graph signals diffused on sample graphs of increasing size. The diffusion outputs have been normalized by $n$.}
  \label{fig:delta_y}
\end{figure}

In this section, we present three numerical experiments to illustrate the results of Thms. \ref{thm:wft} through \ref{thm:non_disc2}. In the first, we sample graph signals from a Gaussian Markov Random Field (GMRF) on ER, SBM and random geometric sensor networks and compare the output of a graph diffusion process as the number of sensors increases. In the second, we compare the WFT of pollutant dispersion signals drawn from the same model on two geometric graphs corresponding to pollution sensor networks in different cities. Finally, in Sec. \ref{sbs:movie} a linear graph filter is optimized to predict movie ratings on a small user network and is then applied to a large one.

\subsection{GMRF diffusion (\textbf{S1})} \label{sbs:gmrf}

In this experiment, we simulate a GMRF measured and diffused on different sensor networks to analyze convergence of the filter $\bbH(\bbS)=\bbS$ in networks of growing size. A graph signal $(\bbG,\bbx)$ is a GMRF on $\bbG$ if $\bbx \sim \ccalN(\bbmu_\bbx, \bbSigma_\bbx)$ and $\bbSigma_\bbx$ is given by \cite{gama2019ergodicity}
\begin{equation} \label{eqn:gmrf_cov}
\bbSigma_\bbx = |a_0|^2(\bbI - a\bbS)^{-1}[(\bbI - a\bbS)^{-1}]^\Hr
\end{equation}
where the covariance matrix is calculated after sampling $\bbG$ from a random graph model for the sensor network, from which we obtain $\bbS$. Three graphons are considered. They are an Erd\"os-R\'enyi (ER) [cf. Fig. \ref{er_graphon}], a stochastic block model (SBM) [cf. Fig. \ref{sbm1}], and a soft random geometric graph [cf. Fig. \ref{exp}]. Their expressions are presented in Table \ref{tbl:W}.

To compare the diffusion outcomes of graph and graphon signals, we first need to define a graphon signal equivalent of the GMRF. We work with its approximation, which is obtained by approximating the graphons as matrices $\bbS_\bbW$. These matrices are calculated by evaluating $\bbW(u_i,u_j)$ on $10^4 \times 10^4$ regularly spaced points of the unit square. Then, the graphon GMRF is obtained by sampling $\bbx_\bbW \in \reals^{10^4}$ from the zero-mean multivariate Gaussian with covariance matrix given by \eqref{eqn:gmrf_cov} for $\bbS=\bbS_\bbW$. 

In order to observe convergence, we compare the outcome of the diffusion of the graphon GMRF with the outcome of the diffusion of a $n$-node graph signal sampled from it for increasing $n$.
This is done by uniformly sampling points $\{u_i\}_{i=1}^n$ from the unit line and generating graphs $\bbG_n$ where the edges $(i,j)$ are Bernoulli random variables with success probability $\bbW(u_i,u_j)$, i.e, $[\bbS_n]_{ij} = [\bbS_n]_{ji} \sim \mbox{Bernoulli}(\bbW(u_i,u_j))$. The graph signals $\bbx_n$ are obtained by interpolating $\bbx_\bbW$ at each $u_i$.

We calculate the diffused graph signals $\bby_n = \bbS_n\bbx_n$ and interpolate the approximation of the diffused graphon signal $\bby_\bbW = \bbS_\bbW\bbx_\bbW$ at $\{u_i\}_{i=1}^n$, then compare them by computing their norm difference for increasing values of $n$. The average normalized norm difference is shown in Fig. \ref{fig:delta_y} for $100$ realizations of the graphon GMRF $\bbx_\bbW$. We observe that, for all graphon models, the norm differences decrease with $n$. This indicates that the vertex response of $\bbH(\bbS) = \bbS$ converges as the graphs $\bbG_n$ grow, as expected from Thm. \ref{thm:non_disc}.

\subsection{Spectral analysis of air pollution on sensor networks (\textbf{S2})} 

The objective of this experiment is to compare the spectral representations of air pollution signals collected at the nodes of two distinct sensor networks of same size to illustrate GFT convergence (Thm. \ref{thm:wft}). This problem can be interpreted as comparing the spectra of graph pollution data in two cities, for instance, New York and Philadelphia.  The air pollution sensor networks are modeled as soft random geometric graphs \cite{penrose2016connectivity} where, given nodes $i$ and $j$ and their coordinates $(x_i,y_i)$ and $(x_j,y_j)$, the probability of connecting $i$ and $j$ is
\begin{equation} \label{eqn:srgg}
p(i,j) \propto \exp \bigg(-\beta \sqrt{(x_i-x_j)^2 + (y_i-y_j)^2}\bigg)\ .
\end{equation}
Fixing the $x$ coordinate at $x_i = x_j = x$ and normalizing $y$ as $u = y/y_{\mbox{{\tiny max}}}$, we can rewrite $p(i,j)$ to fit the expression of the graphon $\bbW(u_i,u_j) = \exp (-\beta\sqrt{(u_i-u_j)^2})$. 

In the cross-wind direction and at fixed altitude, the simplest model for air pollution dispersion is a Gaussian on the distance to the source of pollution in the cross-wind direction. Having fixed $x$, we assume the cross-wind direction to be $y$. The air pollution dispersion model is then
\begin{equation*}
\bbs(y) \propto \exp \bigg(-\dfrac{(y-y_{\mbox{{\tiny source}}})^2}{2\sigma^2}\bigg),
\end{equation*}
where $s(y)$ is the concentration of pollutants at the coordinate $y$ and the variance $\sigma^2$ represents the cross-wind mixing \cite[Chapter 9]{arya1999air}. If we assume $y_{\mbox{{\tiny source}}} = 0$ and once again normalize $y$ as $u = y/y_{\mbox{{\tiny max}}}$,  this dispersion model can be interpreted as a signal $X(u) \propto \exp(-u^2 / 2 \sigma^2)$ on the graphon associated with the soft random geometric graph model of the sensor networks.

\begin{figure}
  \centering
  \includegraphics[width=0.92\columnwidth]{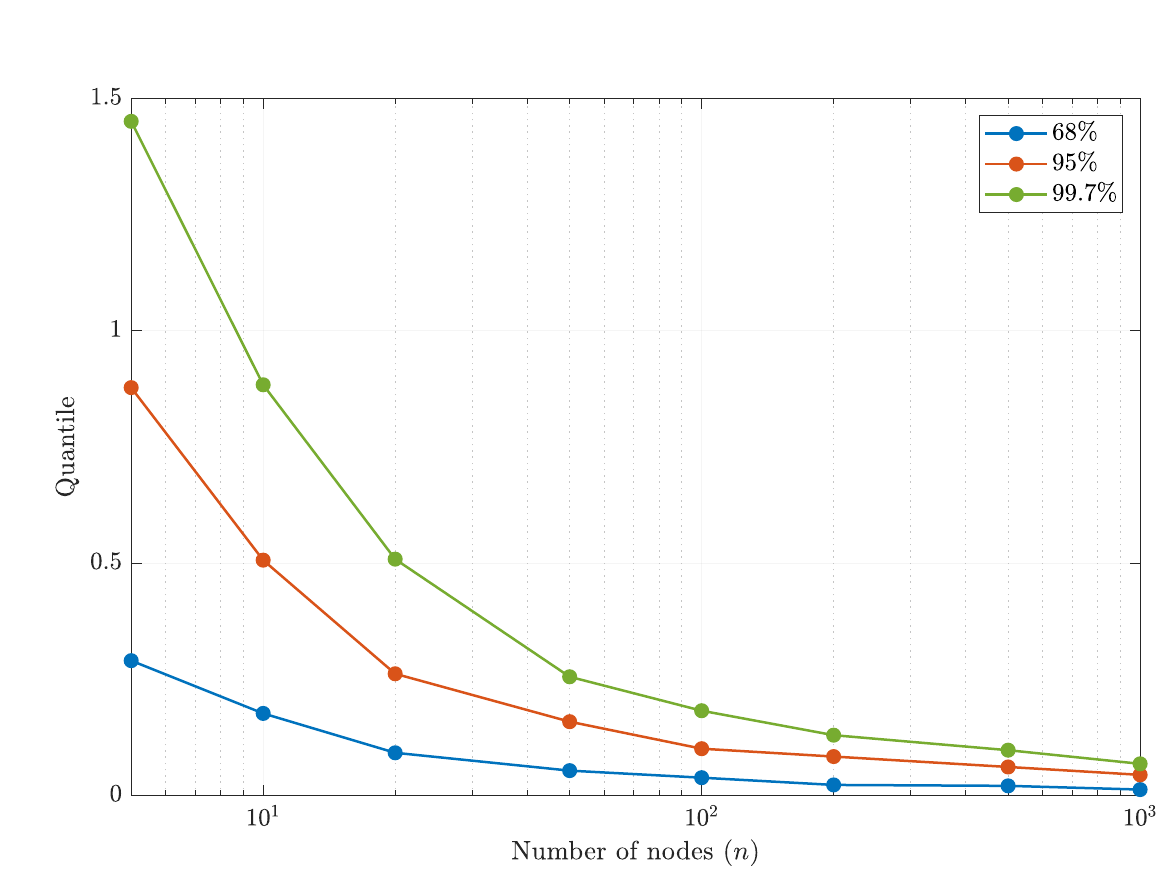}
  \caption{Quantiles ($68\%, 95\%, 99.7\%$) of the minimum normalized difference between GFTs of air pollution signals on graphs drawn from the same geometric model ($\bbG_1$ and $\bbG_2$) for $n=5,10,20,50,100,200,500,1000$, over 50 iterations for each $n$.}
  \label{fig:sensor_net}
\end{figure}

\begin{table}
\centering
\caption{Expression of $\bbW(u_i,u_j)$ for the different graphon models in Sec. \ref{sbs:gmrf}.}
\begin{tabular}{l|c} \hline
Model & Expression of $\bbW(u_i,u_j)$ \\ \hline
ER & $ = 0.4$ for all $u_i,u_j$ \\ \hline
SBM & $
    = 
\begin{cases}
    0.8 ,& \text{if } u_i, u_j \leq 0.5 \text{ or }u_i, u_j > 0.5\\
    0.2,              & \text{otherwise}
\end{cases}
$ \\ \hline
Geom. & $= \exp (-\beta (u_i-u_j)^2 )$, $\beta = 2.3$ \\ \hline
\end{tabular}
\label{tbl:W}
\end{table}

For multiple values of $n$ and using coordinates $\{u_i^{(1)}\}_{i=1}^n$ and $\{u_i^{(2)}\}_{i=1}^n$ sampled uniformly at random from the unit line, we sample two distinct $n$-node graphs  $\bbG_1$ and $\bbG_2$ from \eqref{eqn:srgg}.
In each of these graphs, the graph signals are the pollutant concentrations $[\bbs_1]_i = s(u^{(1)}_i)$ and $[\bbs_2]_i = s(u^{(2)}_i)$. 
We then compute the GFTs $\hbs_1$ and $\hbs_2$, and sort them to find the minimum norm difference $\min \|\hbs_1 - \hbs_2\|$ over different permutations of the labels of these graphs.
After repeating the experiment 50 times for each $n$ in $n = 5, 10, 20, 50, 100, 200, 500, 1000$, we graph the $68\%$, $95\%$ and $99.7\%$ quantile curves of the GFT norm difference (normalized by $\|\hbs_1\|$) in Fig. \ref{fig:sensor_net}. All confidence intervals shrink consistently around the mean as $n$ increases, indicating that the GFTs of the air pollution signals in $\bbG_1$ and $\bbG_2$ indeed converge as expected from Thm. \ref{thm:wft}.

\subsection{Movie rating prediction via user-based graph filtering (\textbf{S3})} \label{sbs:movie}

Given $U$ users and $M$ movies, movie rating prediction consists of completing a $U \times M$ incomplete rating matrix by predicting the ratings users would give to movies that they have not yet rated. We interpret this problem as a GSP problem by considering movie ratings (i.e., the columns of the rating matrix) to be graph signals on a network connecting similar users. A number of graph-based models for movie rating prediction have been proposed in the literature \cite{weiyu18-movie,ruiz19-inv,monti17-movie}. We consider one of the methods in \cite{weiyu18-movie}, which completes the rating matrix by first solving an optimization problem to obtain the optimal coefficients of a linear graph filter, and then applying it to the graph signals corresponding to each movie's rating vector on the user network. 
Our objective is to calculate this graph filter in subnetworks corresponding to small cohorts of users, and observe how well it generalizes when applied to the full user network.

\begin{figure}
  \centering
  \includegraphics[width=0.44\columnwidth]{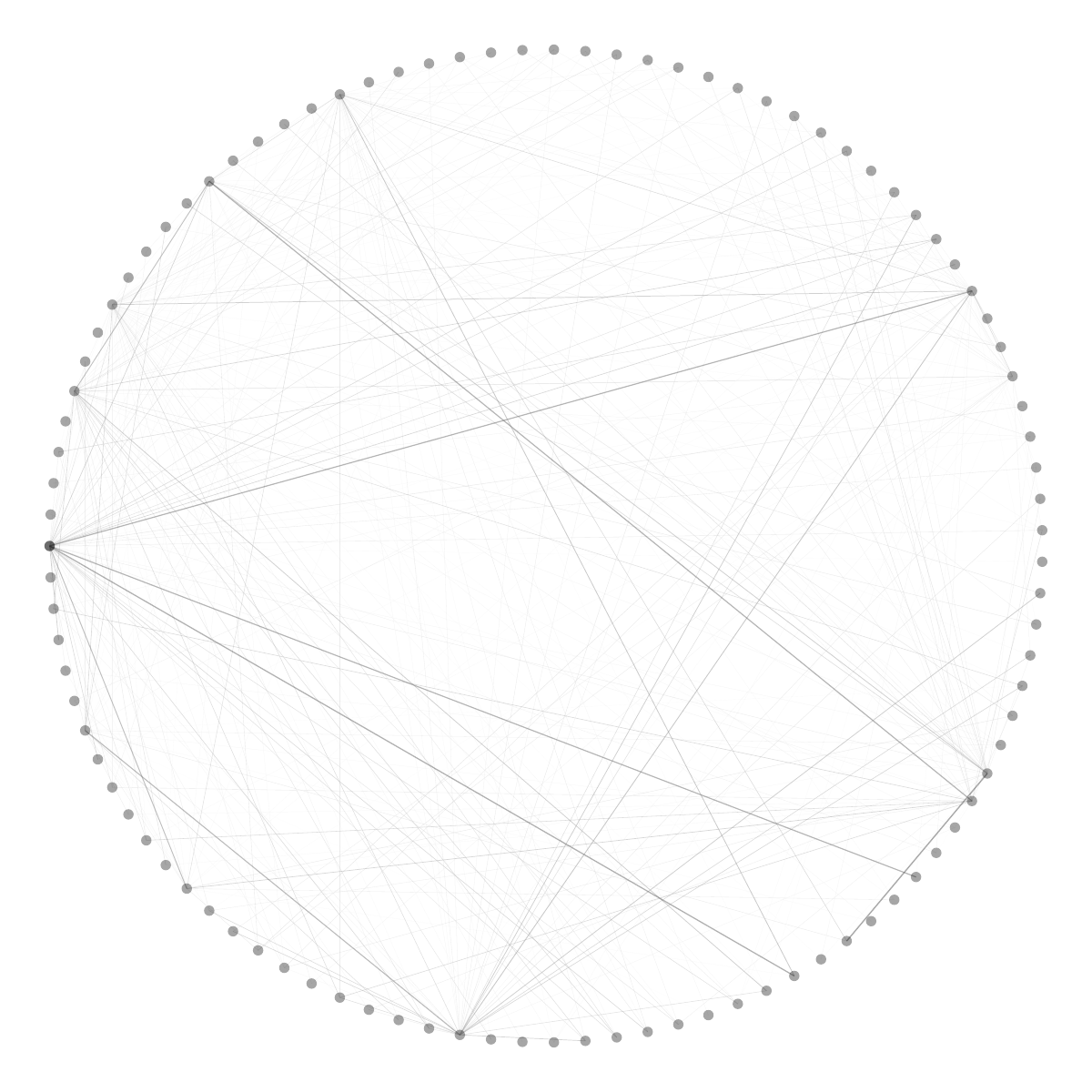}
  \includegraphics[width=0.44\columnwidth]{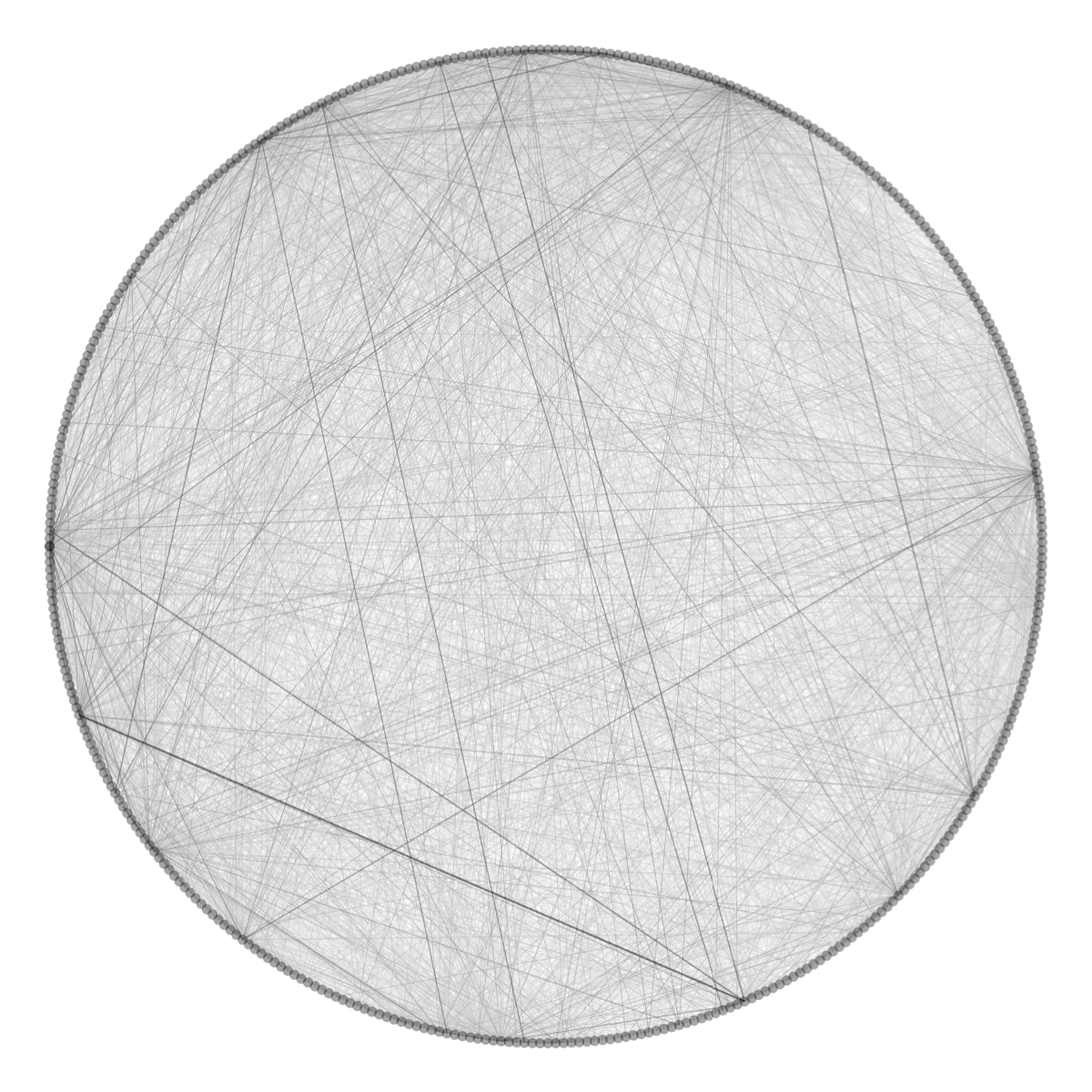}
  \caption{User networks built from the ratings of $100$ (left) ad $400$ (right) users in the MovieLens 100k dataset. The signals on these graphs correspond to the ratings given to the movie ``Toy Story''. The darker the node, the higher the rating, and the darker the edge, the higher the rating difference between the endnodes.}
  \label{fig:movie}
\end{figure}

The dataset we use is the MovieLens 100k dataset \cite{harper16-movielens}, which contains 100,000 ratings by $U=943$ to $M=1582$ movies. The user similarity network is built from the data  by computing pairwise correlations from ratings given by each pair of users to movies that they both have rated and, then, keeping only the top-40 nearest neighbors to each user. Although these are networks built from real data, i.e, to which we cannot attribute a common generative model or graphon, the goal of this section is to illustrate how our results can be implicitly  observed even in graphs that are not related by a common probability model, but that are ``similar'' in some other empirical or statistical sense. This is illustrated in Fig. \ref{fig:movie}, where user networks with $100$ and $400$ users are depicted. Even if the user network on the right has 4 times more users than the one on the left, we can see that the large-scale structure of these networks is similar.

The coefficients of filters with $K=1$, $2$ and $3$ filter taps are optimized on networks of size $50, 100, 200, 400, 600, 800$ and $943$ nodes. We then compare the RMSE obtained by predicting ratings using the filters calculated on the smaller networks and the filters calculated on the full user network. The relative RMSE differences and the base RMSE (obtained from the filter calculated on the full user network) are shown in Table \ref{table:movie}. For a network with $n$ users, the reported RMSE difference corresponds to that of the average among filters trained on $\lfloor 943/n \rfloor$ different networks. Users were picked at random. We observe that, for all $K$, the RMSE difference gets steadily smaller as the network size increases. In particular, for $K=1$ and $K=3$ the relative RMSE difference is less than $1\%$ for filters obtained on networks with under half the number of total users in the dataset.

%% file: conclusions-graphon.tex

\begin{table}[t]
 \caption{Relative RMSE difference for rating prediction based on $K=1,2,3$ filters obtained on $50, 100, 200, 400, 600$ and $800$-user networks, with respect to the base RMSE of the same filters obtained on the full $943$-user network. }
\begin{adjustbox}{width=\columnwidth,center}
\centering
 \begin{tabular}{c|cccccc|c} 
 \hline
  {} & \multicolumn{6}{c|}{Number of users} & {}\\
 K & $50$ & $100$ & $200$ & $400$ & $600$ & $800$ & Base\\
 \hline
 $1$ & $9.70\%$ & $4.70\%$ & $1.90\%$ & $0.45\%$ & $0.17\%$ & $0.04\%$ & $\mathbf{0.77}$\\ 
 $2$ & $22.30\%$ & $20.47\%$ & $14.42\%$ & $5.48\%$ & $2.22\%$ & $0.37\%$ & $\mathbf{0.72}$\\ 
 $3$ & $28.17\%$ & $13.58\%$ & $3.47\%$ & $0.32\%$ & $0.41\%$ & $-0.12\%$ & $\mathbf{0.65}$\\ 
 \hline
 \end{tabular}
 \end{adjustbox}
 \label{table:movie}
\end{table}


We have proposed a novel graphon signal processing framework which simplifies the analysis of signals and the design of filters on very large and dynamic networks. This framework introduces graphon signals, the graphon Fourier transform and LSI graphon filters. We have shown that graphon filters and the WFT are the limit objects of graph filters and of the GFT. These results justify transferring signal analysis methods and information processing systems from graphs to graphons or between graphs associated with the same graphon. GFT and graph filter convergence were demonstrated in two experiments involving graphs drawn from the same graphon, and, in a third experiment, we illustrate how graph filter behavior can be \textit{transferred} even in situations where graphs are built from model-free data and can only empirically or statistically be said to belong to the same ``class''.

%% file: app-graphon.tex


\section{Proof of Lemma 2} \label{appx:lemma2}

The proof follows by direct computation. For $j \in \ccalL$,
\begin{align*}
	&(T_{\bbW_\bbG}\varphi_j)(u) = \int_0^1 \bbW_\bbG(u,v) \varphi_j(v) dv
	\\
	{}&= \sqrt{n} \mbI\left( u \in I_k \right) \int_0^1 [\bbS]_{k\ell} [\bbv_j]_{k} \times 
		\mbI\left( v \in I_\ell \right) dv
	\\
	{}&= \sqrt{n} \mbI\left( u \in I_k \right) \sum_{\ell = 1}^{n} [\bbS]_{k\ell} [\bbv_j]_{k} \int_{I_\ell} dv = \frac{[\bbS \bbv_j]_{k}}{n} \times \sqrt{n} \mbI\left( u \in I_k \right) \\
	 &= \frac{\lambda_j(\bbS)}{n}
			\bigg[ [\bbv_j]_{k} \times \sqrt{n}\mbI\left( u \in I_k \right) \bigg]= \lambda_j(T_{\bbW_\bbG}) \varphi_j(u)
		\text{.}
\end{align*}
If $j \notin \ccalL$, then $\langle{\varphi_j},{\varphi_k}\rangle = 0$ for all $k \in\ccalL$. In this case, we can trivially write $(T_{\bbW_\bbG}\varphi_j)(u) = 0 = \lambda_j(T_{\bbW_\bbG}) \varphi_j(u)$. Note that since the $\bbv_k$ are orthonormal, so are the $\{\varphi_k(T_{\bbW_\bbG})\}$ and therefore a basis completion $\{\varphi_j\}$ can always be obtained. To conclude, compute for $j \in \ccalL$
\begin{align*}
	[\hat{X}_\bbG]_j &= \int_0^1 \varphi_j(v) X_\bbG(v) dv
	\\
	{}&= \sqrt{n} \int_0^1 [\bbv_j]_{\ell} [\bbx]_{\ell}
		\times \mbI\left( v \in I_\ell \right) dv
	\\
	{}&= \sqrt{n} \sum_{\ell = 1}^{n} [\bbv_j]_{\ell} [\bbx]_{\ell} \int_{I_\ell} dv
		= \dfrac{\bbv_j^\Tr \bbx}{\sqrt{n}} = \dfrac{[\hbx]_j}{\sqrt{n}}
		\text{.}
\end{align*}
If $j \notin \ccalL$, recall that since the $\{\bbv_j\}$ form a basis of $\reals^{n}$, we can write $\bbx = \sum_{k \in \ccalL} c_k \bbv_k$. Hence,
\begin{align*}
	[\hat{X}_\bbG]_j &= \int_0^1 \varphi_j(v) X_\bbG(v) dv
	\\
	{}&= \int_0^1 [\bbx]_{\ell} \times \mbI\left( v \in I_\ell \right)\varphi_j(v) dv
	\\
	{}&= \int_0^1 \sum_{k \in \ccalL} c_k [\bbv_k]_{\ell} \times \mbI\left( v \in I_\ell \right)\varphi_j(v) dv
	\\
	{}&= \dfrac{1}{\sqrt{n}} \sum_{k \in \ccalL} c_k \int_0^1 \varphi_k(v) \varphi_j(v) dv = 0
		\text{.} \qed
\end{align*}

\section{Proof of Lemma 3 and Lemma 4} \label{appx:lemma3}

To prove Lemma \ref{T:eigenvectors_convergence}, we first repeat Lemma \ref{eigenvalue_conv} below.

\begin{customlemma}{4}[Eigenvalue convergence]
Let $\{\bbG_n\}$ be a sequence of graphs with eigenvalues $\{\lambda_j(\bbS_n)\}_{j \in \mbZ \setminus \{0\}}$, and $\bbW$ a graphon with eigenvalues $\{\lambda_j(T_\bbW)\}_{j \in \mbZ \setminus\{0\}}$. Assume that, in both cases, the eigenvalues are ordered by decreasing order of absolute value and indexed according to their sign. If $\{\bbG_n\}$ converges to $\bbW$, then, for all $j$
\begin{equation} \label{eqn:eigenvalue_conv2}
\lim_{n \to \infty} \dfrac{\lambda_j(\bbS_n)}{n} = \lim_{n \to \infty} \lambda_j(T_{\bbW_{\bbG_n}}) = \lambda_j(T_\bbW) \ .
\end{equation}
\end{customlemma}
\begin{proof}
The proof is essentially the one for \cite[Thm. 6.7]{borgs2012convergent}, but we reproduce it here using our notation.
Recall that since the sequence $\{\bbG_n\}$ converges to $\bbW$, the density of homomorphisms for any motif also converges. The result then follows by choosing a homomorphism connected to the eigenvalues of their induced operators, namely the $k$-cycle $\bbC_k$. Indeed, notice that for any graphon $\bbW^\prime$ and $k \geq 2$, we have, by definition, that $t(\bbC_k,\bbW^\prime) = \sum_{j \in \mbZ \setminus \{0\}} \lambda_j(T_{\bbW^\prime})^k$. Hence,
\begin{equation}\label{E:k_cycle}
	\lim_{n \to \infty} \sum_{j \in \mbZ \setminus \{0\}} \lambda_j(T_{\bbW_n})^k =
		\sum_{j \in \mbZ \setminus \{0\}} \lambda_j(T_\bbW)^k
		\text{, for } k \geq 2
\end{equation}
where $T_{\bbW_n}=T_{\bbW_{\bbG_n}}$.
It now suffices to show that \eqref{E:k_cycle} implies $\lambda_j(T_{\bbW_n}) \to \lambda_j(T_\bbW)$.

We start by bounding the eigenvalues of any graphon $\bbW^\prime$ in terms of its density of homomorphisms. In particular, for $k = 4$ we obtain that
\begin{align*}
	\sum_{j = 1}^m \lambda_j(T_{\bbW^\prime})^4
		\leq \sum_{j \in \mbZ \setminus \{0\}} \lambda_j(T_{\bbW^\prime})^4 = t(\bbC_4,\bbW^\prime)
	\Rightarrow
	\\
	\lambda_m(T_{\bbW^\prime}) \leq \left[ \frac{t(\bbC_4,\bbW^\prime)}{m} \right]^{1/4}
		\text{ and}
	\\
	\sum_{j = -m}^{-1} \lambda_j(T_{\bbW^\prime})^4
		\leq \sum_{j \in \mbZ \setminus \{0\}} \lambda_j(T_{\bbW^\prime})^4 = t(\bbC_4,\bbW^\prime)
	\Rightarrow
	\\
	\lambda_{-m}(T_{\bbW^\prime}) \geq -\left[ \frac{t(\bbC_4,\bbW^\prime)}{m} \right]^{1/4}
		\text{.}
\end{align*}
Since $t(\bbC_4,\bbW_n)$ is a convergent sequence, it has a bound $B$ \cite{borgs2012convergent}, which implies that
\begin{equation}\label{E:eigenvalue_bound}
	|{\lambda_j(T_{\bbW_n})}| \leq \left( \frac{B}{|j|} \right)^{1/4}
		\text{, for all } j \in \mbZ \setminus \{0\}
		\text{.}
\end{equation}
\blue{Note that for $k \geq 5$, we can take the limit in \eqref{E:k_cycle} term-by-term since, as $|{\lambda_j(T_{\bbW_n})^k}| \leq (B/|j|)^{k/4}$ and the series $\sum_i (B/|j|)^{k/4}$ is convergent for $k > 4$, $\sum_{j \in \mbZ\setminus\{0\}}|{\lambda_j(T_{\bbW_n})^k}|$ also converges.} Hence, from \eqref{E:k_cycle}, we have
\begin{equation}\label{E:limit_termbyterm}
	\lim_{n \to \infty} \sum_{j \in \mbZ \setminus \{0\}} \lambda_j(T_{\bbW_n})^k
		= \sum_{j \in \mbZ \setminus \{0\}} \zeta_j^k
		= \sum_{j \in \mbZ \setminus \{0\}} \lambda_j(T_\bbW)^k
\end{equation}
for $k \geq 5$, where $\zeta_j^k = \lim_{n \to \infty} \lambda_j(T_{\bbW_n})^k$.

To conclude, we proceed by induction over an ordering of the sequence of eigenvalues $\lambda_j(T_\bbW)$, namely over $j_\ell$, $\ell = 1,2,\dots$, such that $|{\lambda_{j_1}(T_\bbW)}| \geq |{\lambda_{j_2}(T_\bbW)}| \geq \dots \geq |{\lambda_{j_\ell}(T_\bbW)}|$. Suppose that $\zeta_{j_\ell} = \lambda_{j_\ell}(T_\bbW)$ for $\ell < \ell^*$ and let $\lambda_{j_{\ell^*}}(T_\bbW)$ be of multiplicity $a$ and appear $b$ times in the sequence $\{\zeta_j\}$ and $-\lambda_{j_{\ell^*}}(T_\bbW)$ be of multiplicity $a^\prime$ and appear $b^\prime$ times in $\{\zeta_j\}$. The identity in \eqref{E:limit_termbyterm} then reduces to
\begin{align*}
	\left[ b + (-1)^k b^\prime \right]
		+ \sum_{\ell > \ell^*} \left( \frac{\zeta_{j_{\ell}}}{\lambda_{j_{\ell^*}}(T_\bbW)} \right)^k =
	\\
	\left[ a + (-1)^k a^\prime \right]
		+ \sum_{\ell > \ell^*} \left( \frac{\lambda_{j_{\ell}}(T_\bbW)}{\lambda_{j_{\ell^*}}(T_\bbW)} \right)^k
		\text{, for } k \geq 5
		\text{,}
\end{align*}
where we divided both sides by $\lambda_{j_{\ell^*}}(T_\bbW)^k$. Due to the ordering of the $\lambda_{j_\ell}$, for $k \to \infty$ through the even numbers we get $b + b^\prime = a + a^\prime$ and through the odd numbers we get $b - b^\prime = a - a^\prime$. Immediately, we have that $a = a^\prime$ and $b = b^\prime$, so that $\zeta_{j_{\ell^*}} = \lambda_{j_{\ell^*}}$. Although this argument assumes $\zeta_{j_\ell} < \lambda_{j_{\ell^*}}$ for all $\ell > \ell^*$, applying the same procedure to an ordering of the sequence $\{\zeta_j\}$ yields the same conclusion.
\end{proof}

We will also require the following well known result about the perturbation of self-adjoint operators. For $\sigma$ a subset of the eigenvalues of a self-adjoint operator $T$, define the spectral projection~$E_T(\sigma)$ as the projection onto the subspace spanned by the eigenfunctions relative to the eigenvalues in $\sigma$.

\begin{proposition}\label{T:davis_kahan}
Let $T$ and $T^\prime$ be two self-adjoint operators on a separable Hilbert space $\ccalH$ whose spectra are partitioned as $\sigma \cup \Sigma$ and $\omega \cup \Omega$ respectively, with $\sigma \cap \Sigma = \emptyset$ and $\omega \cap \Omega = \emptyset$. If there exists $d > 0$ such that $\min_{x \in \sigma,\, y \in \Omega} |{x - y}| \geq d$ and $\min_{x \in \omega,\, y \in \Sigma}|{x - y}| \geq d$, then
\begin{equation}\label{E:davis_kahan}
	\vertiii{E_T(\sigma) - E_{T^\prime}(\omega)} \leq \frac{\pi}{2} \frac{\vertiii{{T - T^\prime}}}{d}
\end{equation}
\end{proposition}

\begin{proof}
See \cite{seelmann2014notes}.
\end{proof}

Lastly, we need two results related to the graphon norm. The first is Lemma \ref{cut_norm_conv}, which states that if a sequence of graphs converges to a graphon in the homomorphism density sense, it also converges in the cut norm \eqref{eqn:cut_norm_def}. 
The second, here presented as Prop. \ref{T:norm_equivalence}, is due to \blue{\cite[Thm. 11.57]{lovasz2012large}} and bounds the $L^2$-induced norm of the graphon operator by is cut norm.

\begin{proposition}\label{T:norm_equivalence}
Let $T_\bbW$ be the operator induced by the kernel $\bbW$. Then, $\|{\bbW}\|_\square \leq \vertiii{{T_\bbW}} \leq \sqrt{8 \|{\bbW}\|_\square}$.
\end{proposition}

\begin{proof}
The lower bound is obtained directly from the first inequality in \cite[Lemma 8.11]{lovasz2012large} and the first inequality in \cite[Lemma E.6]{jansen}. For the upper bound, combine the second inequality in \cite[Lemma 8.11]{lovasz2012large} and the second inequality in Lemma \cite[Lemma E.6]{jansen} to write
\begin{equation*}
  \vertiii{{T_\bbW}}_{2,2}
     \leq \sqrt{2 \vertiii{{T_\bbW}}_{\infty,1} }
              \leq \sqrt{2 \big(4 \|{\bbW}\|_{\square}\big)}.
\end{equation*}
Group terms in the last inequality to conclude the proof. 

Note that~\cite[Lemma E.6]{jansen} is written in terms of the ``type-2 cut norm'' $\|\bbW \|_{\square,2}$~\cite[Equation (4.3)]{jansen}. It is easy to see that $\|\bbW \|_{\square,2}$ is an equivalent definition of the $\infty,1$-operator norm, i.e., $\|\bbW \|_{\square,2} = \vertiii{{T_\bbW}}_{\infty,1}$~\cite[Remark 4.2]{jansen}.
\end{proof}

We can now proceed with the proof of our lemma:

\begin{proof}[Proof of Lemma \ref{T:eigenvectors_convergence}]
For~$j \in \ccalC$, let~$\sigma = \lambda_j(T_\bbW)$, $\Sigma = \{\lambda_i(T_\bbW)\}_{i \neq j}$, $\omega = \lambda_j(T_{\bbW_n})$, and $\Omega = \{\lambda_i(T_{\bbW_n})\}_{i \neq j}$ in Prop. \ref{T:davis_kahan} to get
\begin{equation}\label{E:davis_kahan_1}
	\vertiii{{E_j - E_{j,n}}} \leq \frac{\pi}{2} \frac{\vertiii{{T_{\bbW_n} - T_\bbW}}}{d_{j,n}}
\end{equation}
where $E_j$ and $E_{j,n}$ are the spectral projections of $T_\bbW$ and $T_{\bbW_n}$ with respect to their $j$-th eigenvalue and
\begin{multline*}
	d_{j,n} = \min \big( |{\lambda_j - \lambda_{j+1}(T_{\bbW_n})}|,
		|{\lambda_j - \lambda_{j-1}(T_{\bbW_n})}|, \\
		|{\lambda_{j+1} - \lambda_{j}(T_{\bbW_n})}|,
		|{\lambda_{j-1} - \lambda_{j}(T_{\bbW_n})}|
	\big)
		\text{,}
\end{multline*}
where we omitted the dependence on $\bbW$ by writing $\lambda_j = \lambda_j(T_\bbW)$.

Fix $\epsilon > 0$. From Lemma \ref{eigenvalue_conv}, we know we can find $n_1$ such that $|{d_{j,n} - \delta_j}| \leq \delta_j/2$ for all $n > n_1$, where 
\begin{equation*}
\delta_j = \min\big( |{\lambda_j - \lambda_{j+1}}|, |{\lambda_j - \lambda_{j-1}}| \big)\ .
\end{equation*} 
Since $\bbW$ is non-derogatory, $\delta_j > 0$. Additionally, the cut norm convergence of graphon sequences (Lemma \ref{cut_norm_conv}) together with Prop. \ref{T:norm_equivalence} implies there exists $n_2$ such that $\vertiii{{T_{\bbW_n} - T_\bbW}} \leq \epsilon \delta_j/\pi$. Hence, for all $n > \max(n_1,n_2)$ it holds from \eqref{E:davis_kahan_1} that
\begin{equation}\label{E:convergence_calA}
	\vertiii{{E_j - E_{j,n}}} \leq \frac{\pi}{2} \frac{\epsilon \delta_j/\pi}{\delta_j/2} = \epsilon
		\text{.}
\end{equation}
Since $\epsilon$ is arbitrary, \eqref{E:convergence_calA} proves that the projections onto the eigenfunctions of the same eigenvalue converge. I.e., the eigenfunction sequence $\varphi_j(T_{\bbW_n})$ itself converges weakly. 
Because the norms of the $\varphi_j(T_{\bbW_n})$ and $\varphi_j(T_{\bbW})$ are always equal to one, in this case weak convergence also implies strong convergence.
To see this, note that $\|\varphi_j(T_{\bbW_n})-\varphi_j(T_\bbW)\|^2$ can be written as
\begin{align*}
\|&\varphi_j(T_{\bbW_n})-\varphi_j(T_\bbW)\|^2 \\
&= \langle \varphi_j(T_{\bbW_n})-\varphi_j(T_\bbW), \varphi_j(T_{\bbW_n})-\varphi_j(T_\bbW) \rangle \\
&= \langle  \varphi_j(T_{\bbW_n}), \varphi_j(T_{\bbW_n})-\varphi_j(T_\bbW) \rangle \\
& \quad \quad- \langle  \varphi_j(T_{\bbW}), \varphi_j(T_{\bbW_n})-\varphi_j(T_\bbW) \rangle \\
&= \|\varphi_j(T_{\bbW_n})\|^2 - 2\langle \varphi_j(T_{\bbW_n}), \varphi_j(T_{\bbW})\rangle + \|\varphi_j(T_{\bbW})\|^2 \\
& \quad \quad \to \|\varphi_j(T_{\bbW_n})\|^2 - 2\langle \varphi_j(T_{\bbW}), \varphi_j(T_{\bbW})\rangle + \|\varphi_j(T_{\bbW})\|^2 \\
& \quad \quad = \|\varphi_j(T_{\bbW_n})\|^2 - \|\varphi_j(T_{\bbW})\|^2 = 1-1 = 0
\end{align*}
where the sixth line follows from weak convergence of the $\varphi_j(T_{\bbW_n})$ to $\varphi_j(T_{\bbW})$.

To proceed, let us apply Prop. \ref{T:davis_kahan} to the subspace spanned by the remaining eigenfunctions with indices not in $\ccalC$. \blue{Let $\sigma = \{\lambda_i(T_\bbW)\}_{i \notin \ccalC}$, $\Sigma = \{\lambda_i(T_\bbW)\}_{i \in \ccalC}$, $\omega = \{\lambda_i(T_{\bbW_n})\}_{i \notin \ccalC}$, and~$\Omega = \{\lambda_i(T_{\bbW_n})\}_{i \in \ccalC}$ in \eqref{E:davis_kahan} to get
\begin{equation}\label{E:davis_kahan_2}
	\vertiii{{E^\prime - E^\prime_{n}}} \leq \frac{\pi}{2} \frac{\vertiii{{T_{\bbW_n} - T_\bbW}}}{d_{n}}
		\text{,}
\end{equation}
where $E^\prime$ and $E^\prime_n$ are the projections onto the subspaces given by $\ccalS = \text{span}\left( \{\varphi_i(T_\bbW)\}_{i \notin \ccalC} \right)$ and $\ccalS_n = \text{span}\left( \{\varphi_i(T_{\bbW_n})\}_{i \notin \ccalC} \right)$ respectively. From Prop. \ref{T:davis_kahan}, the denominator $d_n$ must satisfy~$d_n \leq \min_{i \notin \ccalC, j \in \ccalC}|\lambda_i(T_{\bbW_n})-\lambda_{j}(T_\bbW)| = d^{(1)}$ and~$d_n \leq \min_{i \notin \ccalC, j \in \ccalC}|\lambda_i(T_{\bbW})-\lambda_{j}(T_{\bbW_n})| = d^{(2)}$. For~$j \in \ccalC$, we have~$|\lambda_{j}(T_\bbW)| \geq c$ and so~$d^{(1)} \geq \min_{i \notin \ccalC}c - |\lambda_i(T_{\bbW_n})|$. As for $d^{(2)}$, there exists~$n_0$ such that~$d^{(2)} \geq \min_{i \notin \ccalC}c - |\lambda_i(T_{\bbW})|$ for~$n > n_0$ because $\lambda_j(T_{\bbW_n}) \to \lambda_j(T_\bbW)$ for all $j$ from Lemma \ref{eigenvalue_conv}. Thus, for $n > n_0$ Prop.~\ref{T:davis_kahan} holds with~$d_n$ given by
\begin{align*}
\begin{split}
d_n \leq  \min \left[\min_{i \notin \ccalC} {c - |\lambda_i(T_{\bbW_n})|},\ 
\min_{i \notin \ccalC}{c - |\lambda_i(T_{\bbW})|}\right]
\end{split}
\end{align*}
which is satisfied by~$d_n=\inf_{i \notin \ccalC} c - |\lambda_i(T_{\bbW_n})|$. 
Since the graphon $\bbW$ is non-derogatory, there exists an $n_1$ such that $d_n > 0$ for all $n > \max(n_0,n_1)$ and we can use the same argument as above to obtain that $E^\prime_{n} \to E^\prime$ in operator norm. The quantity $d_n$ is illustrated in Fig. \ref{fig:eigenvalues_graphon}.}

To see how this implies that for all $i \notin \ccalC$ the function $\varphi_i(T_{\bbW_n})$ converges weakly to a function in the subspace $\ccalS$---which we denote $\red{\Psi_i}$---, let $\Phi \in \ccalS^\perp$. Then,
\begin{align*}
|\langle \varphi_i(T_{\bbW_n}), \Phi \rangle | &= |\langle E_n'  \varphi_i(T_{\bbW_n}), \Phi \rangle | \\
&= |\langle E_n'  \varphi_i(T_{\bbW_n}), \Phi \rangle - \langle E'  \varphi_i(T_{\bbW_n}), \Phi \rangle|
\end{align*}
where the last equality holds because $ \langle E'  \varphi_i(T_{\bbW_n}), \Phi \rangle = 0$ due to $\Phi \in \ccalS^\perp$. From the linearity of inner products, this can be rewritten as
\begin{align*}
|\langle \varphi_i(T_{\bbW_n}), \Phi \rangle | &= |\langle E_n'  \varphi_i(T_{\bbW_n}) - E'  \varphi_i(T_{\bbW_n}), \Phi \rangle| \\
&= |\langle (E_n' - E')  \varphi_i(T_{\bbW_n}), \Phi \rangle|
\end{align*}
and, applying Cauchy-Schwarz,
\begin{align*}
|\langle \varphi_i(T_{\bbW_n}), \Phi \rangle | \leq \|E_n' - E'\| \|\Phi \| .
\end{align*}
Taking the limit on both sides of the inequality, we get
\begin{align*}
\lim_{n\to\infty}{|\langle \varphi_i(T_{\bbW_n}), \Phi \rangle |} \leq \|\Phi \| \lim_{n\to\infty}{\|E_n' - E'\| } = 0 .
\end{align*}
Hence, $\varphi_i(T_{\bbW_n)}$ converges weakly to a $\Psi_i$ that is perpendicular to elements of $\ccalS^\perp$, i.e., $\Psi_i \in \ccalS$.

\end{proof}


\section{The WSO is a Hilbert-Schmidt Operator} \label{appx:HS}

\begin{lemma} \label{lemma:HS} The graphon operator $T_\bbW$ is a Hilbert-Schmidt operator.
\end{lemma}
\begin{proof} Since $T_\bbW$ is an integral operator with kernel $\bbW$, to prove that it is Hilbert-Schmidt it suffices to show that the $L_2$ norm of $\bbW$ is finite. Because $0 \leq \bbW(u,v) \leq 1$ for all $u,v \in [0,1]$, it holds that
\begin{equation*}
\|\bbW\|_{L_2} = \int_0^1 \int_0^1 |\bbW(u,v)|^2 du dv \leq 1 \text{.}
\end{equation*}
Hence, the $L_2$ norm of $\bbW$ is finite and $T_\bbW$ is Hilbert-Schmidt with Hilbert-Schmidt norm $\|T_\bbW\|_{\mbox{\scriptsize HS}} = \|\bbW\|_{L_2}$.
\end{proof}

\section{The space of non-derogatory graphons is dense} \label{appx:derogatory}
\begin{proposition}\label{T:nonderogatory_density}
The set of operators induced by non-derogatory graphons is dense in the space of linear, compact, self-adjoint operators with respect to the $L^2$-induced norm.
\end{proposition}

\begin{proof}

This is a direct consequence of the fact that every compact, self-adjoint operator is the limit of a sequence of finite rank operators. To see why this is the case, recall that the eigenfunctions $\{\varphi_i\}$ form an orthonormal basis of $L^2([0,1])$ \cite[Chapter 28, Thm. 3]{lax02-functional}. Hence, since $\bbW \in L^2([0,1]^2)$, the induced $T_\bbW$ has finite $L^2$-norm and the sequence $\sum_{i \in \mbZ \setminus \{0\}} |{\langle{T_{\bbW}X},{\varphi_i}\rangle}|^2$ is convergent and can be arranged so that for every $\epsilon > 0$, there exists $n_0$ such that
\begin{equation}\label{E:small_tail}
	\sum_{|{i}| > n} |{\langle{T_{\bbW}X},{\varphi_i}\rangle}|^2 \leq \frac{\epsilon^2 \|{X}\|}{2}
		\text{, for all } n > n_0 \text{.}
\end{equation}

Fix a graphon $\bbW$. We now show that for any $\epsilon > 0$, there exists a non-derogatory graphon $\bbW^\prime$ such that $\vertiii{{T_\bbW - T_{\bbW^\prime}}} \leq \epsilon$. To do so, define the graphon $\bbW^n$ through its operator as in
\begin{equation*}
	T_{\bbW^n}X = \sum_{|{i}| \leq n}
		\langle{T_{\bbW}X},{\varphi_i}\rangle \varphi_i
		+ \sum_{|{i}| \leq n} \delta_i \varphi_i
		\text{,}
\end{equation*}
where the $\delta_i$ are chosen so that $\lambda_i + \delta_i \neq \lambda_j + \delta_j$ for all $|{i}|,|{j}| \leq n$ and $|{\delta_i}| \leq \epsilon/(2\sqrt{n})$. In other words, the $\delta_i$ are small perturbations chosen to guarantee that $T_{\bbW^n}$ is non-derogatory. Since the $\{\varphi_i\}$ form an orthonormal basis, we obtain that
\begin{align*}
	\vertiii{{T_\bbW - T_{\bbW^n}}}^2 &= \sup_{\|{X}\| = 1} \|{T_\bbW X - T_{\bbW^n}X}\|^2	=
	\\
	{}&\sum_{|{i}| \leq n} \delta_i^2
		+ \sup_{\|{X}\| = 1} \sum_{|{i}| > n} |{\langle{T_{\bbW}X},{\varphi_i}\rangle}|^2
	\\
	{}&\leq \frac{\epsilon^2}{2}
		+ \sup_{\|{X}\| = 1} \sum_{|{i}| > n} |{\langle{T_{\bbW}X},{\varphi_i}\rangle}|^2
		\text{.}
\end{align*}
Using \eqref{E:small_tail} and taking $\bbW^\prime = \bbW^{n_0}$, we conclude that $\vertiii{{T_\bbW - T_{\bbW^\prime}}} \leq \epsilon$.
\end{proof}

\begin{corollary} \label{T:non_derog_graphons_dense}
Non-derogatory graphons are dense in the space of graphons with respect to the cut norm.
\end{corollary}
\begin{proof}
This is due to the fact that the operators induced by non-derogatory graphons are dense in the topology induced by the $L^2$ operator norm on the space of compact, self-adjoint operators, cf. Prop. \ref{T:nonderogatory_density}. Since this topology is equivalent to the one induced by the cut norm, this implies that non-derogatory graphons are also dense in the space of graphons with respect to the cut norm.
\end{proof}

\section{Proof of Proposition 2} \label{appx:prop2}

Let $\lambda^n_i  = \lambda_i(\bbS_n)/n$ denote the normalized eigenvalues of the graphs $\pi_n(\bbG_n)$ (and thus the eigenvalues of $\bbW_{n} = \bbW_{\pi_n(\bbG_n)}$, cf. Lemma \ref{T:induced_graphon}), and $\lambda_i$ the eigenvalues of $\bbW$. Now suppose that $\lambda_i$ and $\lambda_j$, $\lambda_i > \lambda_j$, are any two different eigenvalues of $\bbW$ with multiplicities $m_i$ and $m_j$; and that $\{\lambda_{i_k}^n\}_{k=1}^{m_i}$ and $\{\lambda_{j_\ell}^n\}_{\ell=1}^{m_j}$ are the eigenvalues of $\bbW_{n}$ converging to $\lambda_i$ and $\lambda_j$. 
Replacing $\sigma$ by $\lambda_i$ and $\omega$ by $\{\lambda_{i_k}^n\}$ in Prop. \ref{T:davis_kahan}, we get
\begin{equation*}
\vertiii{E_{T_\bbW}(\lambda_i) - E_{T_{\bbW_{n}}}(\{\lambda_{i_k}^n\})} \leq \dfrac{\pi}{2}\dfrac{\vertiii{T_\bbW - T_{\bbW_{n}}}}{\delta_{ij}}
\end{equation*}
where $\delta_{ij} = \min_{(i,\ell),(j,k)}{\{|\lambda_i-\lambda^n_{j_\ell}|,|\lambda_j-\lambda^n_{i_k}|\}}$. The denominator has limit $\lim_{n\to \infty} \delta_{ij} = \lambda_i - \lambda_j > 0$, so $E_{T_{\bbW_{n}}}(\{\lambda_{i_k}^n\}) \to E_{T_\bbW}(\lambda_i)$ follows immediately from convergence of $\{\bbG_n\}$ to $\bbW$ and Lemma \ref{cut_norm_conv} together with Prop. \ref{T:norm_equivalence}.

\section{Proof of Theorem 4} \label{appx:thm4}

\begin{proof} [Proof of Thm. \ref{thm:non_disc2}]
Once again, we leave the dependence on $\pi_n$ implicit and write the graphon signals induced by $(\pi_n(\bbG_n), \pi_n(\bbx_n))$ as $(\bbW_n, X_n)$.
Recall that the spectral properties of these graph signals are preserved in the induced graphon signals per Lemma \ref{T:induced_graphon}. Without loss of generality, we also consider filters with normalized filter function $\bar{h}(\lambda) = h(\lambda)/ \sup_{\lambda \in [0,1]} |h(\lambda)|$ as in the proof of Thm. \ref{thm:non_disc}.

To prove filter output convergence for sequences of graphs converging to arbitrary (possibly derogatory) graphons, we must separate the convergence analysis between spectral components associated with eigenvalues with multiplicity $m_i = 1$ and eigenvalues with multiplicity $m_i > 1$.
Hence, we write the output graphon signal $(\bbW,Y)$ as $Y = Y^{(1)} + Y^{(2)}$, with
\begin{align} \label{eqn:spectral_proof21}
Y^{(1)} &= \sum_{i \in \ccalM_{=1}} \bar{h}(\lambda_i) \hat{X}(\lambda_i) \varphi_i \quad \text{and}\\
Y^{(2)} &= \sum_{i \notin \ccalM_{=1}} \bar{h}(\lambda_i) E(\lambda_i)X
\end{align}
where $\ccalM_{=1} = \{i \ |\ m_i = 1\}$ and $E(\lambda_i)X$ is the projection of $(\bbW,X)$ onto the eigenspace associated with $\lambda_i$.

As for the graphon signals induced by the graph filter outputs $(\bbG_n,\bby_n)$, denoted $(\bbW_n,Y_n)$, their spectral decomposition is split between eigenvalues converging individually to different eigenvalues of $\bbW$ and eigenvalues that are part of a set converging to a common eigenvalue of $\bbW$. I.e., $Y_n = Y^{(1)}_n + Y^{(2)}_n$ such that
\begin{align} \label{eqn:spectral_proof22}
Y^{(1)}_n &= \sum_{i\in \ccalM_{=1}} \bar{h}(\lambda^n_i) \hat{X}_n (\lambda^n_i)\varphi_i(T_{\bbW_n}) \quad \text{and} \\
Y^{(2)}_n &= \sum_{i \notin \ccalM_{=1}} \bar{h}(\lambda^n_i) E(\lambda_i^n)X_n
\end{align}
where $\lambda^n_i = \lambda_i(T_{\bbW_n})$ and $\ccalM_{=1} = \{i \ |\ \lambda^n_i \to \lambda_j, m_j = 1\}$. 

From Thm. \ref{thm:non_disc}, $Y^{(1)}_n \to Y^{(1)}$ as $n \to \infty$. It remains to show that $Y^{(2)}_n \to Y^{(2)}$. Using \eqref{eqn:spectral_proof21} and \eqref{eqn:spectral_proof22}, we write
\begin{align} \label{eqn:spectral_proof23}
\begin{split}
\|Y^{(2)} - &Y^{(2)}_n\| = \\
&\left\|\sum_{i \notin \ccalM_{=1}} \bar{h}(\lambda_i) E(\lambda_i)X - \sum_{i \notin \ccalM_{=1}} \bar{h}(\lambda_i^n) E(\lambda^n_i)X_n\right\|\ .
\end{split}
\end{align}
These sums can be further split by defining the set $\ccalC = \{i\ |\ i \notin \ccalM_{=1},\ |\lambda_i| \geq c\}$, where $c = (1-|\bar{h}_{0}|)/{L(2\|X\|\epsilon^{-1}+1)}$, $\bar{h}_{0}= \bar{h}(0)$ and $\epsilon > 0$. Explicitly, we can use the triangle inequality to write
\begin{equation} \label{eqn:spectral_proof24}
\begin{split}
    &\left\|\sum_{i \notin \ccalM_{=1}} \bar{h}(\lambda_i) E(\lambda_i)X - \sum_{i \notin \ccalM_{=1}} \bar{h}(\lambda^n_i) E(\lambda^n_i)X_n\right\|\\
    &\leq \left\|\sum_{i \in \ccalC}  \bar{h}(\lambda_i) E(\lambda_i)X - \sum_{i \in \ccalC} \bar{h}(\lambda^n_i) E(\lambda^n_i)X_n \right\|\ \mbox{\textbf{(i)} } \\ 
    &+ \left\|\sum_{i \notin \ccalC}  \bar{h}(\lambda_i) E(\lambda_i)X - \sum_{i \notin \ccalC} \bar{h}(\lambda^n_i) E(\lambda^n_i)X_n \right\|\ \mbox{\textbf{(ii)}}\ .
    \end{split}
\end{equation}
Since $X_n \to X$ in $L^2$ and the projection operators converge in the induced operator norm [cf. Prop. \ref{subspace_conv}], $E(\lambda^n_i)X_n \to E(\lambda_i)X$ in $L^2$ for $i \in \ccalC$. From this result and from Thm. \ref{thm:transfer_fcn}, we conclude that there exists $n_0$ such that, for all $n > n_0$,
\begin{equation} \label{eqn:spectral_proof25}
\left\|\sum_{i \in \ccalC} \bar{h}(\lambda_i) E(\lambda_i)X - \sum_{i \in \ccalC} \bar{h}(\lambda^n_i) E(\lambda^n_i)X_n \right\| < \epsilon\ 
\end{equation}
which gives a bound for \textbf{(i)}.

For \textbf{(ii)}, we can use the filter's Lipschitz property and the Cauchy-Schwarz and triangle inequalities to write
\begin{equation} \label{eqn:spectral_proof26}
\begin{split}
&\left\|\sum_{i \notin \ccalC} \bar{h}(\lambda_i) E(\lambda_i)X - \sum_{i \notin \ccalC} \bar{h}(\lambda^n_i) E(\lambda^n_i)X_n\right\|\ \\
& \leq \left\|\sum_{i \notin \ccalC} (\bar{h}_{0}+Lc) E(\lambda_i)X - \sum_{i \notin \ccalC} (\bar{h}_{0}-Lc) E(\lambda^n_i)X_n\right\|\ \\
&\leq|\bar{h}_{0}|\left\|\sum_{i \notin \ccalC} \left[E(\lambda_i)X - E(\lambda^n_i)X_n\right]\right\| \\
&+ Lc\left\|\sum_{i \notin \ccalC} E(\lambda_i)X\right\| + Lc\left\|\sum_{i \notin \ccalC} E(\lambda^n_i)X_n\right\|\ .
\end{split}
\end{equation}
Observe that $\sum_{i \notin \ccalC} E(\lambda_i)X$ can be written as 
\begin{equation} \label{eqn:identities2.1}
\sum_{i \notin \ccalC} E(\lambda_i)X = X - \sum_{i \in \ccalC} E(\lambda_i)X
\end{equation}
and $\sum_{i \notin \ccalC} E(\lambda^n_i)X_n$ as
\begin{equation} \label{eqn:identities2.2}
\sum_{i \notin \ccalC} E(\lambda^n_i)X_n = X_n - \sum_{i \in \ccalC} E(\lambda^n_i)X_n\ .
\end{equation}
Thus, since $X_n \to X$ and $E(\lambda^n_i)X_n \to E(\lambda_i)X$ for $i \in \ccalC$, there exists $n_1$ such that, for all $n > n_1$,
\begin{align} \label{eqn:spectral_proof27}
\begin{split}
&\left\|\sum_{i \notin \ccalC} E(\lambda_i)X - E(\lambda^n_i)X_n\right\| \\
&\leq \|X_n-X\| + \left\|\sum_{i \in \ccalC} E(\lambda^n_i)X_n -  E(\lambda_i)X\right\| < \epsilon\ .
\end{split}
\end{align}

As for $\|\sum_{i \notin \ccalC} E(\lambda^n_i)X_n\|$, using the identities in \eqref{eqn:identities2.1}-\eqref{eqn:identities2.2} and the triangle inequality we can write
\begin{align*}
\begin{split}
\left\|\sum_{i \notin \ccalC} E(\lambda^n_i)X_n\right\| \leq \left\|X_n - X\right\| + \left\| \sum_{i \notin \ccalC} E(\lambda_i)X\right\| \\
+ \left\|\sum_{i \in \ccalC}  E(\lambda_i)X- \sum_{i \in \ccalC} E(\lambda^n_i)X_n\right\|\ .
\end{split}
\end{align*}
Hence, since $X_n \to X$ and $E(\lambda^n_i)X_n \to E(\lambda_i)X$ for $i \in \ccalC$, we conclude that
\begin{align} \label{eqn:spectral_proof27.5}
\begin{split}
\left\|\sum_{i \notin \ccalC} \hat{X}_n(\lambda^n_i) \varphi_i(T_{\bbW_n})\right\|
\leq \epsilon + \left\|\sum_{i \notin \ccalC} \hat{X}(\lambda_i)\varphi_i\right\| \mbox{ for } n>n_1
\end{split}
\end{align}
and finally, using the Cauchy-Schwarz and triangle inequalities and substituting \eqref{eqn:spectral_proof27} and \eqref{eqn:spectral_proof27.5} in \eqref{eqn:spectral_proof26}, we get
\begin{equation*}
\begin{split}
&\left\|\sum_{i \notin \ccalC} \bar{h}(\lambda_i) E(\lambda_i)X - \sum_{i \notin \ccalC} \bar{h}(\lambda^n_i) E(\lambda^n_i)X_n\right\|\ \\
&\leq (|\bar{h}_{0}|+Lc) \epsilon + 2Lc \left\|\sum_{i \notin \ccalC} E(\lambda_i)X\right\| \\
&\leq (|\bar{h}_{0}| +Lc) \epsilon + 2Lc\|X\| = \epsilon
\end{split}
\end{equation*}
which proves that $\|Y-Y_n\| < 2\epsilon$ for all $n > \max{\{n_0,n_1\}}$. 
\end{proof}